\documentclass[conference]{IEEEtran}
\usepackage{amsmath,amssymb,amsthm,cite,graphicx,array}
\usepackage{graphicx}
\usepackage{float}
\usepackage{caption}
\usepackage{multicol}
\usepackage{mathrsfs}
\usepackage{amsmath}
\usepackage{amssymb}
\usepackage{amsthm}
\usepackage{multicol}
\usepackage{algorithm}
\usepackage[noend]{algpseudocode}
\usepackage{float}
\usepackage{placeins}
\usepackage{epstopdf}
\usepackage{multirow}
\usepackage{subfigure}
\usepackage{enumitem}
\usepackage[utf8]{inputenc}
\floatstyle{ruled}
\newfloat{algorithm}{tbp}{loa}
\providecommand{\algorithmname}{Algorithm}
\floatname{algorithm}{\protect\algorithmname}
\theoremstyle{remark}
\newtheorem{theorem}{Theorem}

\newtheorem{lemma}{Lemma}

\newtheorem{definition}{Definition}

\newtheorem{step}{Step}
\newtheorem{note}{Note}

\theoremstyle{remark}
\newtheorem{example}{Example}
\ifCLASSINFOpdf
\else
\fi

\title{A Field-Size Independent Code Construction for Groupcast Index Coding Problems}

\begin{document}
\author{Mahesh~Babu~Vaddi~and~B.~Sundar~Rajan\\ 
 Department of Electrical Communication Engineering, Indian Institute of Science, Bengaluru 560012, KA, India \\ E-mail:~\{vaddi,~bsrajan\}@iisc.ac.in }
\maketitle
\begin{abstract} 
The length of an optimal scalar linear index code of a groupcast index coding problem is equal to the minrank of its side information hypergraph. The side-information hypergraph becomes a side-information graph for a special class of groupcast index coding problems known as unicast index coding problems. The computation of minrank is an NP-hard problem. There exists a low rank matrix completion method and clique cover method to find suboptimal solutions to the index coding problem represented by a side-information graph. However, both the methods  are NP-hard. The number of computations required to find the minrank depends on the number of edges present in the side-information graph. In this paper, we define the notion of minrank-critical edges in a side-information graph and derive some properties of minrank, which identifies minrank-non-critical edges. Using these properties  we present a method for reduction of the given minrank computation problem into a smaller problem. Also, we give an heuristic algorithm to find a clique cover of the side-information graph by using  some binary operations on the adjacency matrix of the side-information graph. We also give a method to convert a groupcast index coding problem into a single unicast index coding problem. Combining all these results, we construct index codes (not necessarily optimal length) for groupcast index coding problems. The construction technique is independent of field size and hence can  be used to construct index codes over binary field. In some cases the constructed index codes are better than the best known in the literature both in terms of the length of the code and the minimum field size required.
\end{abstract}
\section{Introduction}
\label{sec:Introduction}
\IEEEPARstart{A}n index coding problem \cite{ISCO}, comprises of a transmitter that has a set of $K$ messages $\{ x_1,x_2,\ldots,x_K\}$, and a set of $m$ receivers $\{ R_1,R_2,\ldots,R_m\}$. Each receiver, $R_k=(\mathcal{K}_k,\mathcal{W}_k)$, knows a subset of messages, $\mathcal{K}_k \subset X$, called its \textit{side-information}, and demands another subset of messages, $\mathcal{W}_k \subseteq \mathcal{K}_k^\mathsf{c}$, called its \textit{Want-set}. The transmitter can take cognizance of the side-information of the receivers and broadcast coded messages, called the index code, over a noiseless channel. The objective is to minimize the number of coded transmissions, called the length of the index code, such that each receiver can decode its demanded message using its side-information and the coded messages. 


An index coding problem with no restrictions on want-set and side-information is called a groupcast index coding problem.  Without loss of generality a groupcast index coding problem with $m$ receivers and want-set $\mathcal{W}_k$ for $k \in [1:m]$ can be converted into another groupcast index coding problem with  $\sum_{k \in [1:m]} \mathit{\vert\mathcal{W}_{k}\vert}$ receivers such that every receiver wants exactly one message. A groupcast index coding problem with $K$ messages $\{x_1,x_2,\ldots,x_K\}$ can be represented by a hypergraph $\mathcal{H}$ with $K$ vertices $\{x_1,x_2,\ldots,x_K\}$ and $\sum_{k \in [1:m]} \mathit{\vert\mathcal{W}_{k}\vert}$ number of hyperedges \cite{ECIC}.

Consider a groupcast index coding problem with $K$ messages, $m$ receivers each wanting one message and side information hypergraph $\mathcal{H}$. Let $$\mathbf{e}_k=(\underbrace{0~0~~\ldots~0}_{k-1}~1~\underbrace{0~0~\ldots~0}_{K-k}) \in \mathbb{F}_q^K.$$ The support of a vector $\mathbf{u} \in \mathbb{F}_q^K$ is defined to be the set supp$(\mathbf{u})=\big\{k \in [1:K]: u_k \neq 0\big\}$. Let $\mathbf{E} \subseteq [1:K]$. We denote $\mathbf{u} \lhd \mathbf{E}$ whenever supp$(\mathbf{u}) \subseteq \mathbf{E}$. Then, the $\text{minrank}_q(\mathcal{H}$) over $\mathbb{F}_q$ is defined \cite{ECIC} as $\text{min}\{\text{rank}_{\mathbb{F}_q}(\{\mathbf{v}_k+\mathbf{e}_{k}\}_{k \in [1:m]}:\mathbf{v}_k \in \mathbb{F}_q^K, \mathbf{v}_k \vartriangleleft \mathcal{K}_k\}.$ 
In \cite{ECIC}, it was shown that for any given index coding problem, the length of an optimal scalar linear index code over  $\mathbb{F}_q$ is equal to the $\text{minrank}_q(\mathcal{H})$ of its side-information hypergraph. However, finding the minrank for any arbitrary side-information hypergraph is NP-hard \cite{ECIC}. There exists a low rank matrix completion method to find the rank of a binary matrix which is also  NP-hard \cite{minrank2}.


An index coding problem is unicast \cite{OMIC} if the demand sets of the receivers are disjoint. An index coding problem is called single unicast if the demand sets of the receivers are disjoint and every receiver wants only one message. Any unicast index problem can be equivalently reduced to a single unicast index coding problem (SUICP). In a single unicast index coding problem, the number of messages is equal to the number of receivers. 

Any SUICP with $K$ messages $\{x_1,x_2,\ldots,x_K\}$ can be expressed as a side-information graph $G$ with $K$ vertices $\{x_1,x_2,\ldots,x_K\}$. In $G$, there exists an edge from $x_i$ to $x_j$ if the receiver wanting $x_i$ knows $x_j$. In a unicast index coding problem with $K$ messages and $K$ receivers, the side-information graph has $\sum_{k \in [1:K]} \mathit{\vert\mathcal{K}_{k}\vert}$ number of edges. A matrix $\mathbf{A}=(a_{i,j})$ fits $G$ if $a_{i,i}=1$ for all $i$ and $a_{i,j}$=0 whenever $(i,j)$ is not an edge of $G$. Let $rk_{q}(\mathbf{A})$ denote the rank of this matrix over $\mathbb{F}_q$. The $minrank_{q}(G)$ is defined as 
\begin{align*}
minrank_{q}(G) \triangleq min\{rk_{q}(\mathbf{A}) : \mathbf{A} \ fits  \ G\}.
\end{align*}

In a side-information graph, if receiver $R_i$ knows $x_j$ and receiver $R_j$ knows $x_i$, then the vertices $x_i$ and $x_j$ in the side-information graph are connected with an undirected edge. The undirected edges in the side-information graph contribute towards cliques in the side-information graph. All the receivers which want a message symbol in a clique can be satisfied by one index code symbol which is the XOR of all message symbols present in the clique. All the receivers which wants a message symbol in a cycle of length $k$ can be satisfied by $k-1$ index code symbols.

As finding the minrank of a side-information graph is NP-hard, many researchers have proposed heuristic methods to solve the minrank problem. Birk \textit{et al.} \cite{ISCO} proposed least difference greedy (LDG) clique cover algorithm to find the cliques in the side-information graph. LDG algorithm works by computing all the possible distances between the rows of the fitting matrix. Kwak \textit{et al.} \cite{eldg} proposed extended least difference greedy (ELDG) clique cover algorithm to find the cliques in the side-information graph. ELDG algorithm works by computing all possible distances between the rows and columns of the fitting matrix. ELDG algorithm also gives a method to find directed cycles of length three. Awais \textit{et. al}\cite{cycle2} proposed an algorithm to piggyback a message which is sparsely connected to cycles of length $k$ on the given graph $G$.  

For a side-information graph $G$ with $K$ vertices, the adjacency matrix $\mathbf{A}=(a_{i,j})$ is a binary square matrix of size $K\times K$ with $1$ at the $(i,j)$th position if there exist a directed edge from $x_i$ to $x_j$ and $0$ else. Note that an undirected edge from $x_i$ to $x_j$ implies that there exists two directed edges, one from $x_i$ to $x_j$ and the other from $x_j$ and $x_i$. Hence, in the side-information graph if there exists an undirected edge from $x_i$ to $x_j$, there exist $1$ in $(i,j)$th and $(j,i)$th positions in the adjacency matrix. For any positive integer $n$, the $n$th power of the adjacency matrix $\mathbf{A}$ gives some information about the paths of length $n$ in the graph $G$. The $(i,j)$th entry of the matrix $\mathbf{A}^n$ gives the number of paths of length $n$ from the vertex $x_i$ to $x_j$ \cite{cycle1}. Similarly, The $(k,k)$th entry of the matrix $\mathbf{A}^n$ gives the number of cycles of length $n$, which pass through $x_k$. These properties of adjacency matrix were used in \cite{cycle2} to give an algorithm to piggyback a message which is sparsely connected to cycles of length $n$ on the given side-information graph $G$.

In the given side-information graph $G$ with $K$ vertices, if there exists no cycles of length less than or equal to $n-1$ for some positive integer $n$, the presence of a cycle of length $n$ can be determined by computing $\mathbf{A}^n$. In $\mathbf{A}^n$, if all diagonal elements are zero, this implies that there exists no cycles of length $n$.  

Index coding is motivated by wireless broadcasting applications where the side-information may be a random quantity. Let $p$ be the probability that the receiver $R_k$, $k \in [1:K]$ knows the message $x_j$ as side-information for $j\in [1:K]\setminus k$. Then, $G$ is a random graph with vertices $\{x_1,x_2,\ldots,x_K\}$, such that each edge between any two vertices occurs with probability $p$, independently of all other edges. 
The size of cliques and cycles in random graphs were extensively studied in the literature. Grimmett \textit{et. al.} \cite{randomgraph2} proved that as the number of vertices in $G$ tends to infinity, the size of maximum clique in $G$ would of $\frac{2log_2(K)}{log_2(\frac{1}{p})}$ with probability one. Being the size of the cliques is small (order of $log_2(K)$), heuristic algorithms may be very useful to give a solution to the index coding problem. In this paper, we give a heuristic approach to find the clique cover and a heuristic method to convert the groupcast index coding problem into a unicast index coding problem.  

In a given index coding problem with side-information graph $G$, an edge $e$ is said to be critical if the removal of $e$ from $G$ strictly reduces the capacity region. The index coding problem $G$ is critical if every edge $e$ is critical. Tahmasbi \textit{et al.} \cite{TSG} studied critical graphs and analyzed some properties of critical graphs with respect to capacity region. 

In this paper, we analyze properties of minrank by defining the notion of minrank-critical edges. 
\begin{definition}
In a given index coding problem with side-information graph $G$, an edge $e$ is said to be minrank-critical if the removal of $e$ from $G$ strictly increases the minrank of the graph $G$. 
 An edge $e\in E$ is said to be minrank-non-critical if the removal of $e$ from $G$ does not change the minrank of the graph $G$. 
\end{definition}

The computation of minrank over binary field requires the computation of the rank of $2^{\sum_{k \in [1:K]} \mathit{\vert\mathcal{K}_{k}\vert}}$ number of binary matrices of size $K \times K$. Hence, identification of every minrank-non-critical edge can reduce the number of computations required to compute the minrank by half.

A directed graph $G$ with $K$ vertices is called $\kappa(G)$-partial clique \cite{ISCO} iff every vertex in $G$ knows atleast $(K-1-\kappa)$ messages as side-information and there exits atleast one vertex in $G$ which knows exactly $(K-1-\kappa)$ messages as side-information. For an index coding problem whose side information graph is a $\kappa(G)$-partial clique, maximum distance separable (MDS) code of length $K$ and dimension $\kappa+1$, over a finite field $\mathbb{F}_q$ for $q \geq K$, can be used as an index code. $\kappa(G)$-partial clique method provides a savings of $K-\kappa-1$ transmissions when compared with the naive technique of transmitting all $K$ messages. Tehrani \textit{et. al} in \cite{bipartiate} proposed a partition multicast technique to address the groupcast index coding problem. In the partition multicast, one divides the messages into partitions and consider each partition as a partial clique. The messages are partitioned in such a way that the sum of savings of all partitions are maximized. However, the proposed partition multicast technique is suboptimal and NP-hard and the required field size depends on the number of messages in the partition and the number of messages known to each receiver in the partition.


Throughout we assume a finite field with characteristic 2 and use the XOR operation for convenience. However the results are easily extendable to finite fields with any characteristic.

\subsection{Contributions}
The main contributions of this paper are summarized as follows.
\begin{itemize}
\item We give a method to construct index codes for groupcast index coding problems which is independent of field size. Partition multicast index codes is the best known in the literature for groupcast index coding problems and they do not exist for all fields. We give instances of groupcast index coding problem where the length of index code obtained by using proposed method is less than that of partition multicast. 
\item To give a method to construct index codes for groupcast index coding problem, we develop many tools to address single unicast index coding problems. We define the notion of minrank-critical edges in a side-information graph and derive some properties of minrank, which identify minrank-non-critical edges in a side-information graph. By using the properties of minrank, we give an algorithm to convert the given minrank computation problem into a smaller problem. We give a heuristic algorithm to find a clique cover of the side-information graph. We also give a sub-optimal method to convert a groupcast index coding problem into a single unicast index coding problem.
\end{itemize}

The remaining part of this paper is organized as follows. In Section \ref{sec2}, we derive some properties of the minrank of a side-information graph and give a method to reduce the complexity of the minrank computation problem. In Section \ref{sec4}, we give a method to construct index codes for groupcast index coding problems which works over every finite field. We conclude the paper in Section \ref{sec5}. In the Appendix we give a heuristic algorithm to find the clique cover of side-information graph. 


\section{Properties of minrank of a side-information graph}
\label{sec2}
In this section, we derive some properties of the minrank of the index coding problem. By using the derived properties, we provide a method to identify minrank-non-critical edges of a side-information graph. As the number of computations required to find exact value of the  minrank is exponential in the number of edges present in the side-information graph, identification of every minrank-non-critical edge can reduce the number of computations required to compute the minrank by half. 
\begin{lemma}
\label{lemma2}
Let $G$ be the side-information graph of an SUICP with $K$ messages. Let $G^{(k)}$ be the side-information graph after removing all the incoming and outgoing edges associated with a vertex $x_k$ for any $x_k \in V(G)$. Then, the minrank of $G^{(k)}$ is atmost one greater than the minrank of $G$.  
\end{lemma}
\begin{proof}
The fitting matrix $\mathbf{A}$ of $G$ is a $K \times K$ matrix. Let $G_k$ be the induced subgraph of vertices $V(G)\setminus \{x_k\}$ in $G$. let $\mathbf{A}_k$ be the fitting matrix of $G_k$. The minrank of $G_k$ is defined as  
\begin{align*}
minrank_{q}(G_k) \triangleq min\{rk_{q}(\mathbf{A}_k) : \mathbf{A}_k~\text{fits in} \ G_k\}.
\end{align*}
\noindent
The matrix $\mathbf{A}_k$ is a $(K-1) \times (K-1)$ matrix which can be obtained from $\mathbf{A}$ by removing the row and column corresponding to $x_k$. Thus the minrank of $G_k$ can not be greater than the minrank of $G$. 

The graph $G^{(k)}$ is the union of $G_k$ and the isolated vertex $x_k$. The minrank of the isolated vertex is one and the minrank of a union of disjoint subgraphs is equal to the sum of the minrank of the subgraphs. Thus the minrank of $G^{(k)}$ is atmost one greater than the minrank of $G$.
\end{proof}

\begin{lemma}
\label{lemma50}
Consider the side-information graph $G$ in Fig. \ref{minrankfig2} in which $V(G)=V(G_1)\cup V(G_2) \cup V(G_3)$ and there are no edges between the vertex sets $V(G_1)$ and $V(G_3)$. Then, we have 
\begin{align*}
&\text{minrank}(G_1)+\text{minrank}(G_3) \leq \text{minrank}(G) \\& \leq \text{minrank}(G_1)+\text{minrank}(G_2)+\text{minrank}(G_3).
\end{align*}

\begin{figure}[h]
~~~~~~~~~~~~~~~~~\includegraphics[scale=0.4]{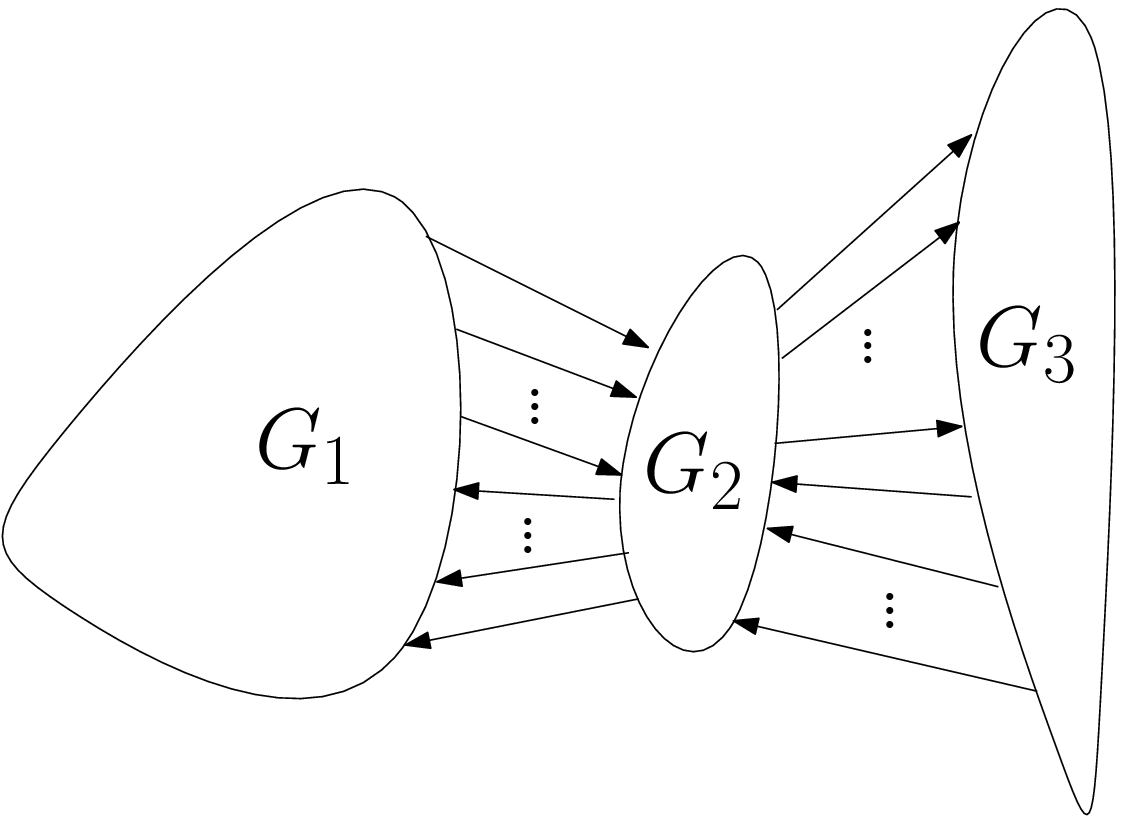}
\caption{}
\label{minrankfig2}
\end{figure}
\end{lemma}
\begin{proof}
The fitting matrix $\mathbf{A}$ of $G$ is a $V(G) \times V(G)$ matrix. Let $G_{13}$ be the induced subgraph of the vertices $V(G_1)\cup V(G_3)$ in $G$. Let $\mathbf{A}_{13}$ be the fitting matrix of $G_{13}$. The minrank of $G_{13}$ is defined as  
\begin{align*}
\text{minrank}(G_{13}) \triangleq min\{\text{rk}(\mathbf{A}_{13}) : \mathbf{A}_{13}~\text{fits in} \ G_{13}\}.
\end{align*}

The matrix $\mathbf{A}_{13}$ can be obtained from $\mathbf{A}$ by removing the rows and columns corresponding to the vertices in $V(G_2)$. Thus the minrank of $G_{13}$ can not be greater than the minrank of $G$. But the graph $G_{13}$ is the disjoint union of the graphs $G_1$ and $G_3$. Hence, we have
\begin{align}
\label{eq4}
\nonumber 
 \text{minrank}(G_{13})& =\text{minrank}(G_1)+\text{minrank}(G_3) \\& \leq \text{minrank}(G).
\end{align}
If we remove all the incoming and outgoing edges to $V(G_2)$ in $G$, then the resulting graph is a disjoint union of $G_1,G_2,G_3$ and hence the minrank of this resulting graph is the sum of the minrank of $G_1,G_2$ and $G_3$. Hence, the minrank of $G$ can not be more than the sum of minrank of $G_1,G_2$ and $G_3$. This along with \eqref{eq4} completes the proof.
\end{proof}
As a special case of Lemma \ref{lemma50} the following lemma is obtained.
\begin{lemma}
Consider the graph $G$ in Fig. \ref{minrankfig1} in which  $V(G)=V(G_1)\cup V(G_2) \cup x_k$ and there are no edges between the vertex sets $V(G_1)$ and $V(G_2)$. Then, we have
\begin{align*}
&\text{minrank}(G_1)+\text{minrank}(G_2) \leq \text{minrank}(G) \\& \leq \text{minrank}(G_1)+\text{minrank}(G_2)+1.
\end{align*}

\begin{figure}[h]
~~~~~~~~~~~~~~~~~\includegraphics[scale=0.4]{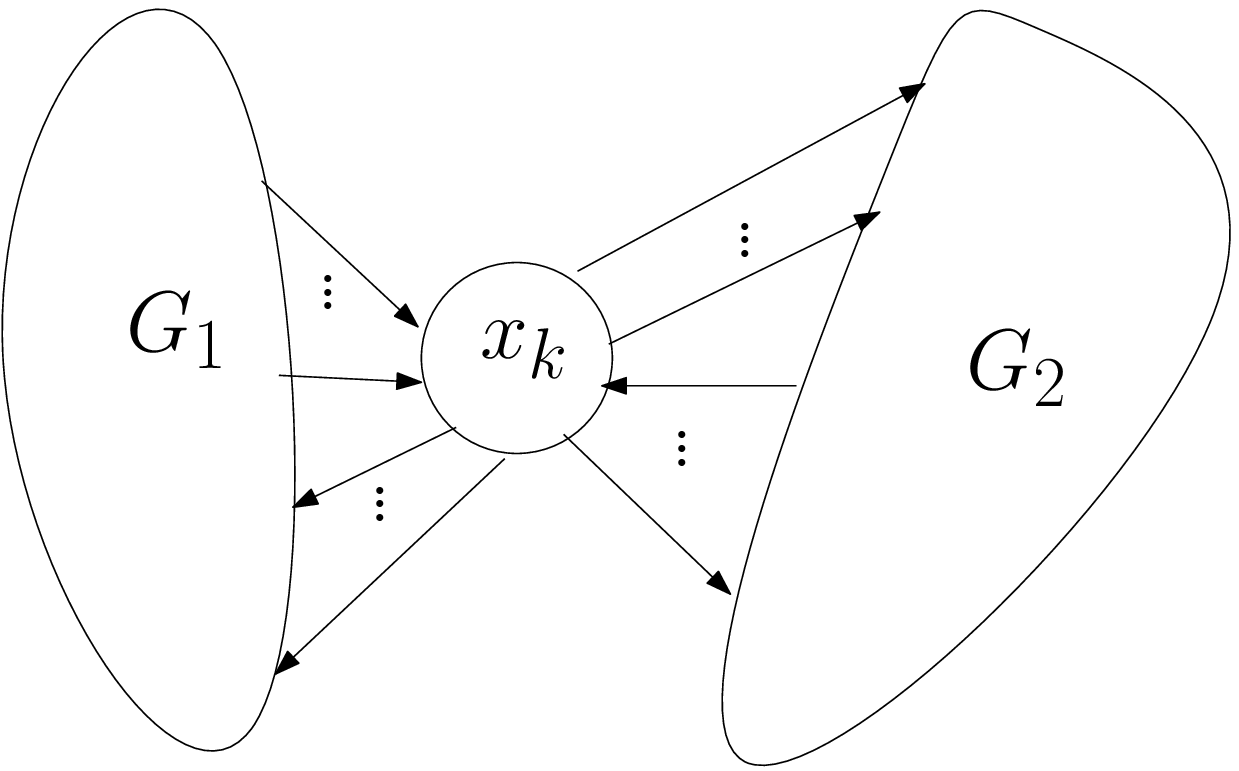}
\caption{}
\label{minrankfig1}
\end{figure}
\end{lemma}

\begin{theorem}
\label{lemma21}
Let $G$ be a side-information graph and $G_k$ be the induced subgraph of $G$ with the vertex set $V(G)\setminus \{x_k\}$ for any $x_k \in V(G)$. If $x_k$ is not present in any directed cycle in $G$ and the minrank of $G_k$ is $m-1$, then the minrank of $G$ is $m$.  
\end{theorem}
\begin{proof}
From Lemma \ref{lemma2}, minrank of $G$ is either $m-1$ or $m$. Consider the case when the minrank of $G$ is $m-1$. That is, the row and the column corresponding to the vertex $x_k$ in the fitting matrix of $G$ are in the span of rows and columns corresponding to the remaining vertices of the graph respectively. Let $L_j$ be the row corresponding to $x_j$ in the fitting matrix of $G$ for any $j \in [1:K]$. Let $L_{k}$ be in the span of $L_{i_1},L_{i_2},\ldots,L_{i_t}$ for some $i_1,i_2,\ldots,i_t \in [1:K]\setminus \{k\}$. That is, 
\begin{align}
\label{lemma3eqn}
L_{k}+L_{i_1}+L_{i_2}+\ldots+L_{i_t}=\textbf{0}
\end{align}
The rows $L_{i_1},L_{i_2},\ldots,L_{t_t}$ and $L_k$ have $1$s in $i_1,i_2,\ldots,i_t$ and $k$th positions respectively ($\mathbf{A}(k,k)=\mathbf{A}(i_1,i_1)=\mathbf{A}(i_2,i_2)=\mathbf{A}(i_t,i_t)=1$). The linear dependence condition in \eqref{lemma3eqn} indicates that there exist atleast one more non zero element in $k,i_1,i_2,\ldots,i_t$ positions in $L_{i_1},L_{i_2},\ldots,L_{t_t}$ and $L_k$ such that the ones in the diagonal positions gets canceled. Note that a non zero element in $j$th position of $L_{s}$ for any $j,s \in [1:t]$ indicates that there is an edge from $x_{s}$ to $x_{j}$ in $G$. This implies that there exists a cycle $x_k \Rightarrow x_{i_1} \Rightarrow x_{i_2} \Rightarrow \ldots \Rightarrow x_{i_t} \Rightarrow x_{k}$. This is a contradiction to our assumption that there exists no cycle through $x_k$. Hence, the rank of $G$ is $m$.
\end{proof}
\begin{lemma}
\label{lemma60}
In the side-information graph $G$, if $x_k$ is not present in any directed cycle in $G$, then all the incoming and outgoing edges from $x_k$ are minrank-non-critical. 
\end{lemma}
\begin{proof}
Let $G_k$ be the induced subgraph with the vertex set $V(G)\setminus \{x_k\}$ in $G$. Let the minrank of $G_k$ be $m-1$. If $x_k$ is not present in any directed cycle in $G$, from Theorem \ref{lemma21}, the minrank of $G$ is $m$. Let $G^{(k)}$ be the graph after removing all the incoming and outgoing edges from $x_k$ in $G$. Hence, the graph $G^{(k)}$ is the union of $G_k$ and isolated vertex $\{x_k\}$. As the minrank of an isolated vertex is one, we have 
\begin{align*}
\text{minrank}(G^{(k)})=\text{minrank}(G_k)+1=m=\text{minrank}(G).
\end{align*}
Hence, removing all incoming and outgoing edges from $x_k$ does not reduce the minrank of $G$ and these edges are minrank-non-critical.
\end{proof}
Let $\mathbf{A}$ be the adjacency matrix of $G$ with $K$ vertices. From the properties of the  adjacency matrix, if $\mathbf{A}^t(k,k)$ is zero, then there exist no cycles in $G$ with length $t$ and which contains $x_k$. Hence the presence of $x_k$ in any directed cycle can be obtained from $\sum_{i=1}^{K} \mathbf{A}^i$. In $\sum_{i=1}^{K} \mathbf{A}^i$, if $(k,k)$th element is zero, this indicates that $x_k$ is not present in any directed cycle in $G$. Hence by computing $\sum_{i=1}^{K} \mathbf{A}^i$, one can identify all the vertices which are not present in any cycle.

\begin{example}
Consider the side-information graph $G$ given in Fig. \ref{minrankfig31}. The adjacency matrix $\mathbf{A}$ of $G$ and $\sum_{k=0}^7\mathbf{A^k}$ are shown below. The $(3,3)$ element in $\sum_{k=0}^7\mathbf{A^k}$ is zero indicates that the vertex $x_3$ is not present in any directed cycle. Hence, from Theorem \ref{theorem1}, all the incoming and outgoing edges from $x_3$ are minrank-non-critical. Let $\mathbf{A}_3$ be the matrix after deleting the third row and third column of $\mathbf{A}$. We have minrank($\mathbf{A}$)=minrank($\mathbf{A}_3$)+1 and we can compute the minrank of $\mathbf{A}$ by computing the minrank of $\mathbf{A}_3$.
\begin{figure}[h]
~~~~~~~~~~~~~~~~~\includegraphics[scale=0.5]{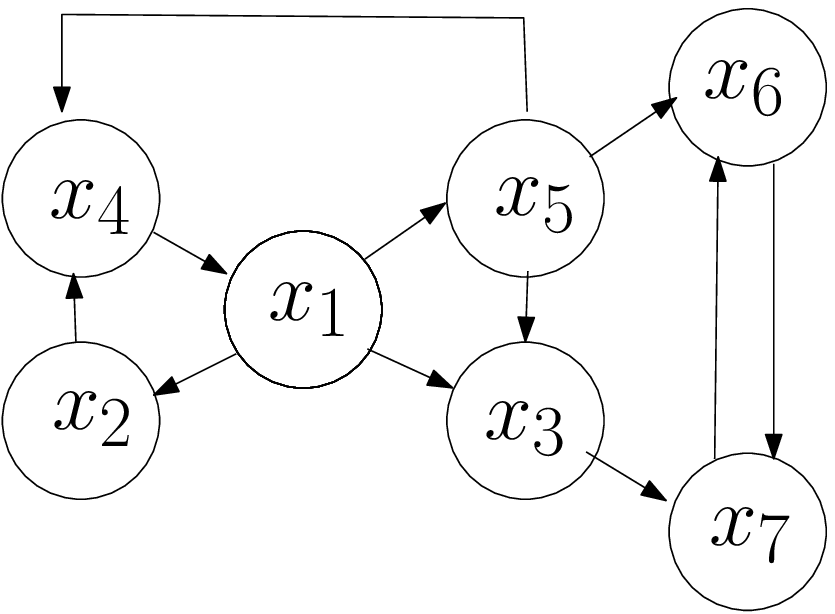}
\caption{}
\label{minrankfig31}
\end{figure}

\begin{small}
\arraycolsep=1pt
\setlength\extrarowheight{-2.0pt}
{
$$\mathbf{A}=\left[\begin{array}{*{20}c}
   0 & 1 & 1 & 0 & 1 & 0 & 0\\
   0 & 0 & 0 & 1 & 0 & 0 & 0\\
   0 & 0 & 0 & 0 & 0 & 0 & 1\\
   1 & 0 & 0 & 0 & 0 & 0 & 0\\
   0 & 0 & 1 & 1 & 0 & 1 & 0\\
   0 & 0 & 0 & 0 & 0 & 0 & 1\\
   0 & 0 & 0 & 0 & 0 & 1 & 0\\
  \end{array}\right],\sum_{k=0}^7\mathbf{A^k}=\left[\begin{array}{*{20}c}
   6 &~ 7~ & 8~ & 6~ & 7~ & 16~ & 14\\
   3 &~ 3~ & 6~ & 7~ & 3~ & 7~ & 7\\
   0 &~ 0~ & 0~ & 0~ & 0~ & 3~ & 4\\
   7 &~ 3~ & 6~ & 6~ & 3~ & 11~ & 12\\
   3 &~ 3~ & 7~ & 7~ & 3~ & 14~ & 13\\
   0 &~ 0~ & 0~ & 0~ & 0~ & 3~ & 4\\
   0 &` 0~ & 0~ & 0~ & 0~ & 4~ & 3\\
  \end{array}\right].$$
}
\end{small}

\end{example}
\begin{lemma}
\label{lemma3}
Let $\mathfrak{C}$ be the set of vertices in any clique of size $t$ in the graph $G$. Let $G_C$ be the side-information graph after removing all the incoming and outgoing edges associated with the $t$ vertices in $\mathfrak{C}$, i.e., $G_C=V(G)\setminus \mathfrak{C}$. Then, the minrank of $G_C$ is atmost one greater than the minrank of $G$.
\end{lemma}
\begin{proof}
Let $G_1$ be the subgraph induced by the $t$ vertices in the clique $\mathfrak{C}$ in $G$ and $G_2$ be the subgraph induced by the remaining vertices in $G$. We have $V(G)=V(G_1)\cup V(G_2)$. The minrank of $G_1$ is one as it is a clique. Hence, we have 
\begin{align*}
\text{minrank}(G) &\leq \text{minrank}(G_1)+\text{minrank}(G_2)=\text{minrank}(G_C)\\&= 1+\text{minrank}(G_2)\leq 1+\text{minrank}(G).
\end{align*}
The last inequality in the above equation follows from the fact that $G_2$ is a subgraph of $G$. This completes the proof.
\end{proof}
\begin{definition}
\label{def1}
Let $G$ be a side-information graph. Let $C_i=\{x_{i_1},x_{i_2},\ldots,x_{i_{\mid C_i \mid}}\}$ and $C_j=\{x_{j_1},x_{j_2},\ldots,x_{j_{\mid C_j \mid}}\}$ be two cliques in $G$. Let $V_R=V(G)\setminus (\{x_{i_1},x_{i_2},\ldots,x_{i_{\mid C_i \mid}}\}\cup \{x_{j_1},x_{j_2},\ldots,x_{j_{\mid C_j \mid}}\})$. We say that the cliques $C_i$ and $C_j$ are cycle-free if there exists atleast two vertices $x_k \in C_i$ and $x_{k^\prime} \in C_j$ such that there is no cycle consisting  of vertices only from a non-trivial subset of $\{x_k,x_{k^\prime}\}$ and any subset of $V_R$. 

\end{definition}

The three examples given below illustrate Definition \ref{def1}.
\begin{example} 
Consider the side-information graph $G_1$ given in Fig. \ref{minrankfig35}. In $G_1$, an undirected edge represents a clique of size two. In $G_1$, there exist two cliques $\{x_1,x_2,x_3\}$ and $\{x_4,x_5\}$. Let $V_R=V(G_1)\setminus (\{x_1,x_2,x_3\} \cup \{x_4,x_5\})=\{x_6,x_7,x_8,x_9,x_{10},x_{11}\}$. In $G_1$, every vertex in the clique $\{x_1,x_2,x_3\}$ is having an out going edge to every vertex in the clique $\{x_4,x_5\}$. The vertex $x_1$ is present in a cycle which comprises of vertices only from the set $V_R$ ($x_1 \rightarrow x_{10} \rightarrow x_{11} \rightarrow x_1$). Similarly,  the vertex $x_2$ is present in a cycle which comprises of vertices only  from the set $V_R$ ($x_2\rightarrow x_8 \rightarrow x_9\rightarrow x_2$). The vertex $x_3$ is not present in any cycle comprising of vertices only from the set $V_R$. In the clique $\{x_4,x_5\}$, both the vertices $x_4$ and $x_5$ are not present in any directed cycle  which comprises of vertices only from the set $V_R$. But, there exists a directed cycle comprising of $x_3$ from clique $\{x_1,x_2,x_3\}$ along with $x_5$ from clique $\{x_4,x_5\}$ and vertices only from the set $V_R$ ($x_3 \rightarrow x_5 \rightarrow x_6 \rightarrow x_7 \rightarrow x_3$). There also exists a directed cycle comprising of $x_3$ from clique $\{x_1,x_2,x_3\}$ along with $x_4$ from clique $\{x_4,x_5\}$ and vertices only from the set $V_R$ ($x_3 \rightarrow x_4 \rightarrow x_6 \rightarrow x_7 \rightarrow x_3$).

Hence, according to Definition \ref{def1}, the clique $\{x_1,x_2,x_3\}$ and the clique $\{x_4,x_5\}$ are not cycle-free. 
\begin{figure}[h]
~~~~~~~~~~~~~~~~~\includegraphics[scale=0.5]{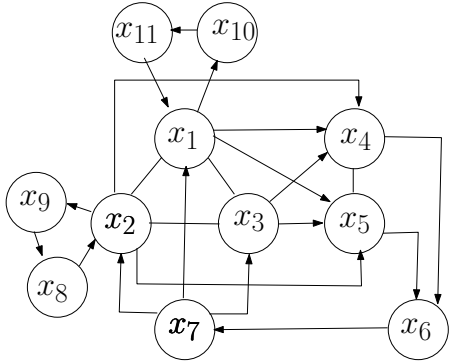}
\caption{Side-information graph $G_1$}
\label{minrankfig35}
\end{figure}
\end{example} 

\begin{example} 
Consider the side-information graph $G_2$ given in Fig. \ref{minrankfig36}. The graph in Fig. \ref{minrankfig36} is same as that of the graph in Fig. \ref{minrankfig35}, except one edge from $x_3$ to $x_5$ is removed. In $G_2$, there exist two cliques $\{x_1,x_2,x_3\}$ and $\{x_4,x_5\}$. In the graph $G_2$, the vertex $x_3$ in the clique $\{x_1,x_2,x_3\}$ is not present in any cycle comprising of vertices only from the set $V_R$.  The vertex $x_5$ in the clique $\{x_4,x_5\}$ is not present in any directed cycle which comprises of vertices only from the set $V_R$. There also does not exist a directed cycle comprising of $x_3$ from the clique $\{x_1,x_2,x_3\}$ along with $x_5$ from the clique $\{x_4,x_5\}$ and vertices only from the set $V_R$.

Hence, according to Definition \ref{def1}, the clique $\{x_1,x_2,x_3\}$ and the clique $\{x_4,x_5\}$ are cycle-free. 
\begin{figure}[h]
~~~~~~~~~~~~~~~~~\includegraphics[scale=0.5]{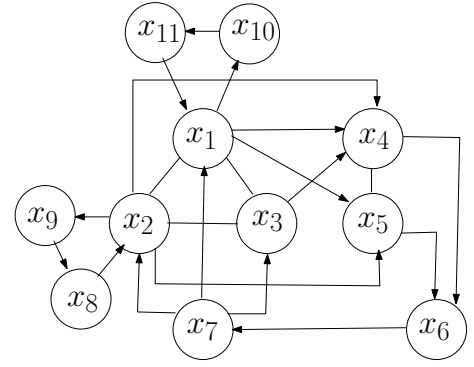}
\caption{Side-information graph $G_2$}
\label{minrankfig36}
\end{figure}
\end{example} 

\begin{example} 
Consider the side-information graph $G_3$ given in Fig. \ref{minrankfig37}. The graph in Fig. \ref{minrankfig37} is same as that of the graph in Fig. \ref{minrankfig35}, except the direction of one edge from $x_6$ to $x_7$ is reversed. In $G_3$, there exist two cliques $\{x_1,x_2,x_3\}$ and $\{x_4,x_5\}$ and every vertex in the clique $\{x_1,x_2,x_3\}$ is having an outgoing edge with every vertex of clique $\{x_4,x_5\}$. In the graph $G_3$, the vertex $x_3$ in the clique $\{x_1,x_2,x_3\}$ is not present in any cycle comprising of vertices only from the set $V_R$.  The vertex $x_5$ in the clique $\{x_4,x_5\}$ is not present in any directed cycle which comprises of vertices only from the set $V_R$. There also does not exist a directed cycle comprising of $x_3$ from the clique $\{x_1,x_2,x_3\}$ along with $x_5$ from the clique $\{x_4,x_5\}$ and vertices only from the set $V_R$.

Hence, according to Definition \ref{def1}, the clique $\{x_1,x_2,x_3\}$ and the clique $\{x_4,x_5\}$ are cycle-free. 
\begin{figure}[h]
~~~~~~~~~~~~~~~~~\includegraphics[scale=0.5]{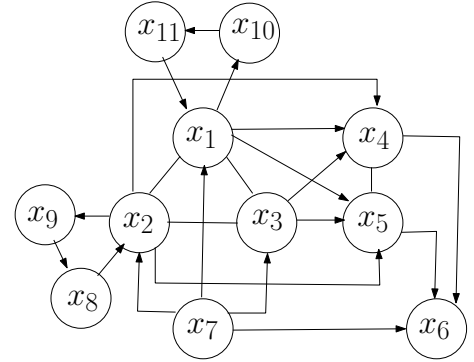}
\caption{Side-information graph $G_3$}
\label{minrankfig37}
\end{figure}
\end{example} 

Theorem \ref{theorem1} given below identifies the minrank-non-critical edges between cliques.

\begin{theorem}
\label{theorem1}
Let $G$ be a side-information graph. Let $C_i=\{x_{i_1},x_{i_2},\ldots,x_{i_{\mid C_i \mid}}\}$ and $C_j=\{x_{j_1},x_{j_2}\ldots,x_{j_{\mid C_j \mid}}\}$ be any two cliques in $G$ that are cycle-free. Then all the edges between $C_i$ and $C_j$ are minrank-non-critical.
\end{theorem}
\begin{proof}
Let the number of edges between $C_i$ and $C_j$ be $\lambda$. We prove that all these $\lambda$ edges are minrank-non-critical. Let $\tilde{G}$ be the side-information graph after deleting $\lambda$ edges between $C_i$ and $C_j$ in $G$. We prove that the minrank of $G$ can not be less than the minrank of $\tilde{G}$. Let $V_R=V(G)\setminus (C_i \cup C_j)$. Let $G_3$ be the induced graph of $V_R$ in $G$. From Lemma \ref{lemma3}, we have 
\begin{align}
\label{eq2}
\text{minrank}(\tilde{G}) \leq \text{minrank}(G_3)+2.
\end{align}

We show that the minrank of $G$ can not be less than $\text{minrank}(G_3)+2$. Let $x_k \in C_i$ and $x_{k^\prime} \in C_j$ be the two vertices such that there is no cycle consists of vertices from non-trivial subset of $\{x_k,x_{k^\prime}\}$ and any subset of $V_R$. Let $G_1$ be the induced graph of the vertices $\{x_k,x_{k^\prime}\}$ in $G$. Let $G_2$ be the induced graph of the vertices $C_i \cup C_j \setminus \{x_k,x_{k^\prime}\}$ in $G$.  We have $V(G)=V(G_1)\cup V(G_2) \cup V(G_3)$.

From Definition \ref{def1} and Lemma \ref{lemma60}, all incoming and outgoing edges from $V(G_1)$ ($\{x_k,x_{k^\prime}\}$) to $V_R$ (vertices of $G_3$) are minrank-non-critical. Hence, the vertices in $G_1$ are connected to $G_3$ via $G_2$. Figure \ref{fig50} is useful to illustrate $G$. From Definition \ref{def1}, there exists no cycle among $x_k$ and $x_{k^\prime}$. Hence, we have minrank($G_1$)=2. From Lemma \ref{lemma50}, we have 
\begin{align}
\label{eq1}
\nonumber
\underbrace{2}_{\text{minrank}(G_1)}&+\text{minrank}(G_3)  \leq \text{minrank}(G) \\& \leq \underbrace{2}_{\text{minrank}(G_1)}+\text{minrank}(G_3)+\text{minrank}(G_2).
\end{align}

Hence, from \eqref{eq2} and \eqref{eq1}, we have 
\begin{align*}
\text{minrank}(\tilde{G})\leq 2+\text{minrank}(G_3) \leq \text{minrank}(G).
\end{align*}
This completes the proof. 

\begin{figure}[h]
~~~~~~~~~~~~~~~~~\includegraphics[scale=0.55]{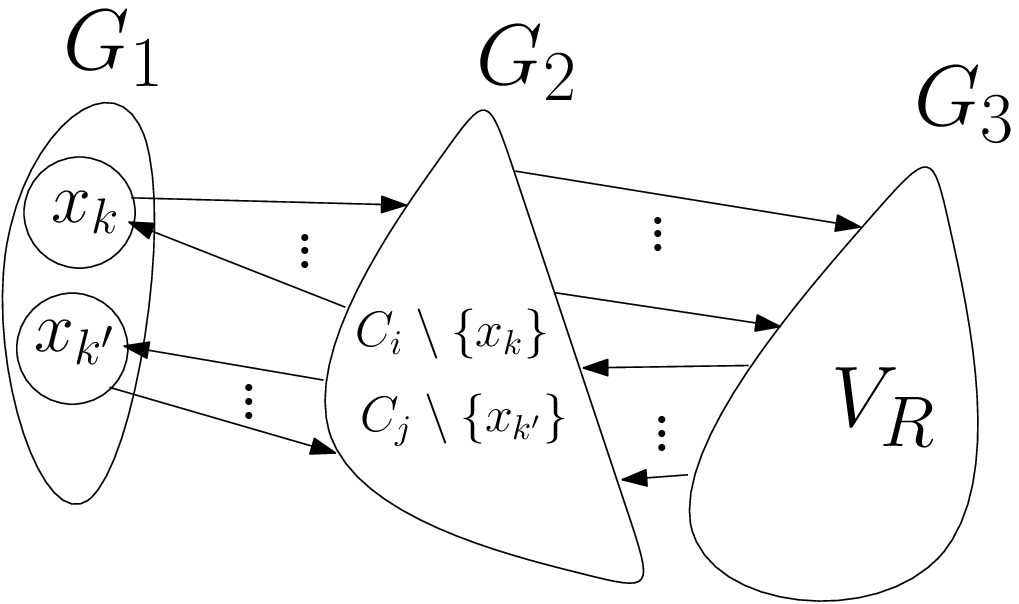}
\caption{}
\label{fig50}
\end{figure}

\end{proof}


Note that Definition \ref{def1} and Theorem \ref{theorem1} are also applicable if clique $C_i$ is a single vertex (trivial clique) or $C_j$ is a single vertex or both $C_i$ and $C_j$ are single vertices.

\begin{theorem}
\label{theorem2}
Let $G$ be a side-information graph with $K$ vertices $\{x_1,x_2,\ldots,x_K\}$. Let $\tilde{G}$ be the graph obtained from $G$ by the following reduction procedure:
\begin{enumerate}
\item[•] Find a set of cliques $\{C_1,C_2,\ldots,C_t\}$ in $G$ such that all the $t$ cliques partition $V(G)$. Note that any vertex is also a trivial clique of size one. If the cliques $C_i$ and $C_j$ are cycle-free, delete all the edges between $C_i$ and $C_j$ for every $i,j \in [1:t]$.
\end{enumerate}
Then,  $
\text{minrank}(G)=\text{minrank}(\tilde{G}).$ 
\end{theorem}
\begin{proof}
Let $C_i$ and $C_j$ be two cycle-free cliques. From Theorem \ref{theorem1}, if $C_i$ and $C_j$ are cycle-free, then all the directed edges between $C_i$ and $C_j$ are minrank-non-critical. Hence, the construction procedure given in the construction of $\tilde{G}$ would not increase the minrank and we have   
$\text{minrank}(\tilde{G})=\text{minrank}(G).$
\end{proof}


Example \ref{ex6} given below illustrates Theorem \ref{theorem2}.
\begin{example} 
\label{ex6}
Consider the side-information graph $G_3$ given in Fig. \ref{minrankfig37}. In $G_3$, there exist two cliques $\{x_1,x_2,x_3\}$ and $\{x_4,x_5\}$ and every vertex in the clique $\{x_1,x_2,x_3\}$ is having an outgoing edge with every vertex of clique $\{x_4,x_5\}$. According to Definition \ref{def1}, the clique $\{x_1,x_2,x_3\}$ and the clique $\{x_4,x_5\}$ are cycle-free. Hence, from Theorem \ref{theorem1}, the six edges from clique $\{x_1,x_2,x_3\}$ to clique $\{x_4,x_5\}$ are minrank-non-critical. Graph $\tilde{G}_3$ after removing the six minrank-non-critical edges is shown in Fig. \ref{minrankfig38}. From Theorem \ref{theorem2}, the minrank of $G_3$ and $\tilde{G}_3$ is same.
\begin{figure}[h]
~~~~~~~~~~~~~~~~~\includegraphics[scale=0.5]{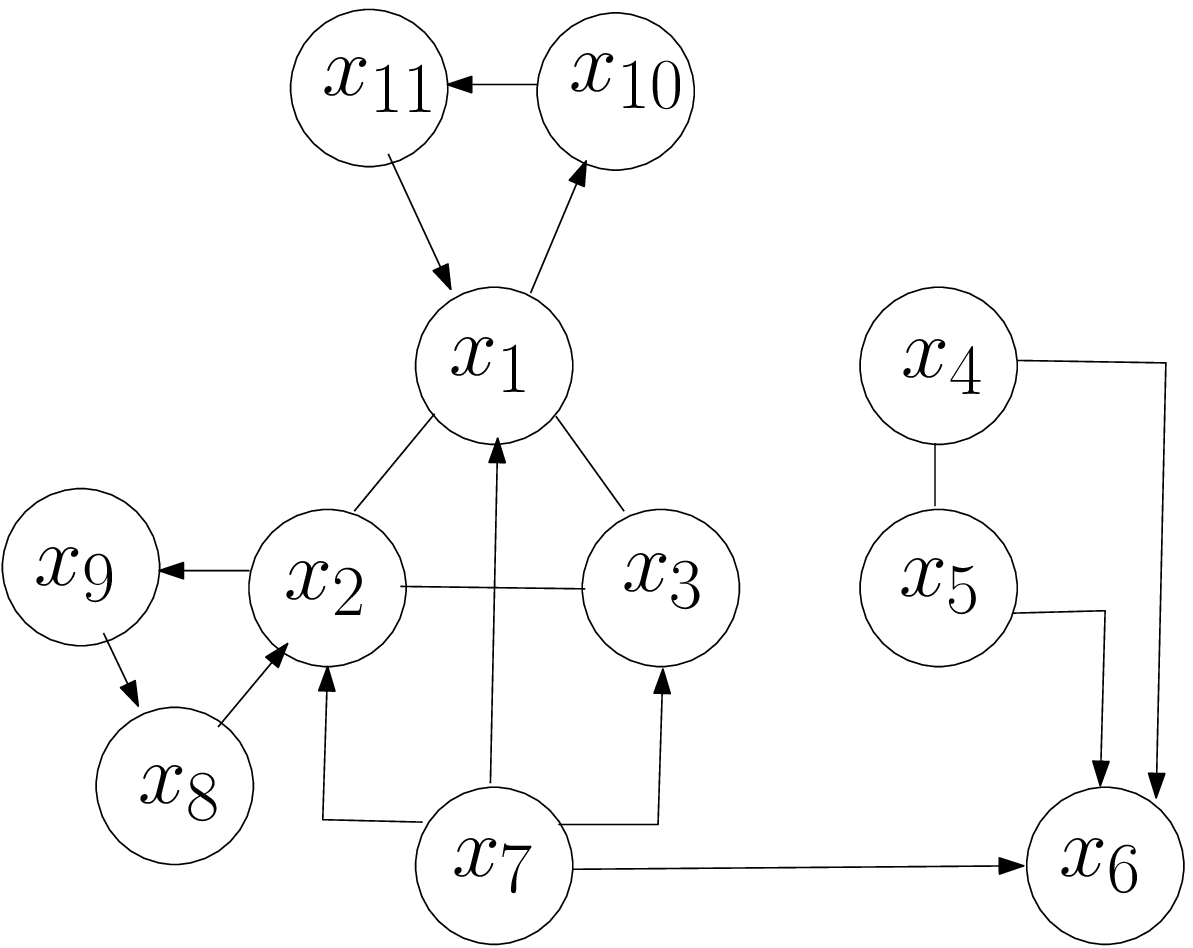}
\caption{}
\label{minrankfig38}
\end{figure}
\end{example} 
Theorem \ref{theorem2} reduces the minrank computation problem into a smaller problem in terms of number of edges (number of vertices remain same after reduction). Construction I given in next subsection reduces the minrank computation problem into a smaller problem in terms of both the number of vertices and the number of edges.

\subsection{A heuristic method to reduce the minrank computation problem}

In the following three steps, we give a heuristic approach to reduce the minrank computation problem into a smaller problem. We refer the following three steps as Construction I in the rest of the paper.

\noindent
{\bf Construction I}

\noindent
{\it Step 1:} Let $\tilde{G}$ be the graph obtained from Theorem \ref{theorem2} with $\{C_1,C_2,\ldots,C_t\}$ being the set of $t$ cliques in $G.$  These $t$ cliques partition $V(\tilde{G})=\{x_1,x_2,\ldots,x_K\}$ (the cliques need not satisfy the cycle-free condition). Note that any vertex is also a trivial clique of size one. Let $G_R$ be the graph obtained from $\tilde{G}$ after the following next two step: 

\noindent
{\it Step 2:}  Let $\{x_{i_1},x_{i_2},\ldots,x_{i_{|C_i|}}\}$ be the vertices in the $i$th clique for $i \in [1:t]$. If $\mid C_i \mid \geq 2$, combine these $|C_i|$ vertices into one new vertex $y_i$. Else, leave the vertex in $C_i$ as it is.

\noindent
{\it Step 3:}Now the number of vertices is equal to the number of cliques in $G$, that is $t$. 
If the number of directed edges from $C_i$ to $C_j$ in $\tilde{G}$ are $|C_i|.|C_j|,$ then introduce a directed edge from $y_i$ to $y_j$ for $i,j \in [1:t].$ Otherwise, there does not exist a directed edge from $y_i$ to $y_j$ for $i,j \in [1:t].$ 
 

~ \\

The reduction of minrank computation problem by using Theorem \ref{theorem2} and Construction I is summarized below.
\begin{align*}
&~~~~~~~G~~ \underbrace{\Rightarrow}_{\text{Theorem}~\ref{theorem2}}~~ \tilde{G}~ \underbrace{\Rightarrow}_{\text{Construction I}}~ G_R \\&
\text{minrank(G) = minrank}{(\tilde{G})}~ \leq~\text{minrank}(G_R).
\end{align*}

In Lemma \ref{lemma121}, we give a sufficient condition when the minrank of the graphs $G$ and $G_R$ are equal. The necessary and sufficient conditions that the side-information graph $G$ need to satisfy such that the minrank of $G$ is equal to the minrank of $G_R$ needs further investigation.

For a graph $G$, the order of an induced acyclic sub-graph formed by removing the minimum number of vertices in $G$, is called Maximum Acyclic Induced Subgraph ($MAIS(G)$). In \cite{ICSI}, it was shown that $MAIS(G)$ lower bounds the minrank of $G$. That is, 
\begin{align}
\label{mr1}
\text{minrank}(G) \geq MAIS(G).
\end{align}

\begin{lemma}
\label{lemma121}
Let $G$ be a side-information graph. Let $\{C_1,C_2,\ldots,C_t\}$ be a set of $t$ cliques in $G$ obtained in Theorem \ref{theorem2}. If every pair of these $t$ cliques are cycle-free, then the minrank of $G$ is equal to the minrank of $G_R$.
\end{lemma}
\begin{proof}
As there exist $t$ cliques in $G_R$ and every clique requires one index code transmission, we have 
\begin{align}
\label{mr2}
\text{minrank}(G_R)\leq t.
\end{align}

Given every pair of $t$ cliques being cycle-free, from Theorem \ref{theorem2}, all the incoming and outgoing edges to every pair of cliques are minrank-non-critical. Let $G^\prime$ be the graph after removing all the incoming and outgoing edges from every pair of cliques. We have 
\begin{align}
\label{mr3}
\text{minrank}(G)=\text{minrank}(G^\prime).
\end{align}

In $G^\prime$, if we choose one vertex from each clique, then all these $t$ vertices form an acyclic induced subgraph.  Hence, we have 
\begin{align}
\label{mr4}
\text{MAIS}(G^\prime) \geq t.
\end{align}

By combining \eqref{mr1}-\eqref{mr4}, we have 
\begin{align*}
t \leq \text{MAIS}(G) & \leq \text{minrank}(G)=\text{minrank}(G^\prime)\\&  \leq \text{minrank}(G_R) \leq t.
\end{align*} 
This completes the proof.
\end{proof}


Lemma \ref{lemma122} given below gives the relation between the index code for $G_R$ and the index code for $G$.
\begin{lemma}
\label{lemma122}
Let $G$ be the side-information graph of a single unicast ICP. Let $G_R$ be the graph obtained from $G$ by using Construction I. An index code $\mathfrak{C}$ for the ICP represented by $G_R$ can be used as an index code for the ICP represented by $G$ after replacing $y_i$ with the XOR of vertices present in $C_i$ for $i \in [1:t]$ ($y_i$ and $C_i$ are defined in Construction I). 
\end{lemma}
\begin{proof}
In the ICP represented by $G$, receiver $R_j$ wants to decode $x_j$ for every $j \in [1:K]$. Let $x_j \in C_i$ for some $i \in [1:t]$. From Construction I, receiver $R_j$ knows all the messages corresponding to out neighbourhood of $y_i$ in $G_R$. Hence, $R_j$ computes $y_i$ from $\mathfrak{C}$ and from $y_i$ it computes $x_j$. 
\end{proof}

The following three examples illustrate Construction I.

\begin{example}
\label{ex98}
Consider the index coding problem represented by the side-information graph $G$ given in Fig. \ref{minrankfig3}. In $G$, the vertices $\{x_1,x_2,x_3\}$ form a clique of size three and the vertices $\{x_4,x_5\}$ form a clique of size two. In the graph $G$, there exists a directed edge from every vertex in the clique $\{x_1,x_2,x_3\}$ to every vertex in the clique $\{x_4,x_5\}$ except an edge from $x_2$ to $x_5$. From Construction I, we can combine the vertices $\{x_1,x_2,x_3\}$ into one vertex and the vertices $\{x_4,x_5\}$ into another vertex in the reduced side-information graph $G_R$. The reduced side-information graph $G_R$ is given in Fig. \ref{minrankfig4}. In this example, the minrank of the graph $G$ is four and the minrank of graph $G_R$ is four. The index code for the index coding problem represented by $G_R$ is $\{y_1,y_4,x_6,x_7\}$ and the index code for the index coding problem represented by $G$ is $\{\underbrace{x_1+x_2+x_3}_{y_1},~\underbrace{x_4+x_5}_{y_4},~x_6,~x_7\}$.
\begin{figure}[h]
~~~~~~~~~~~~~~~~~\includegraphics[scale=0.5]{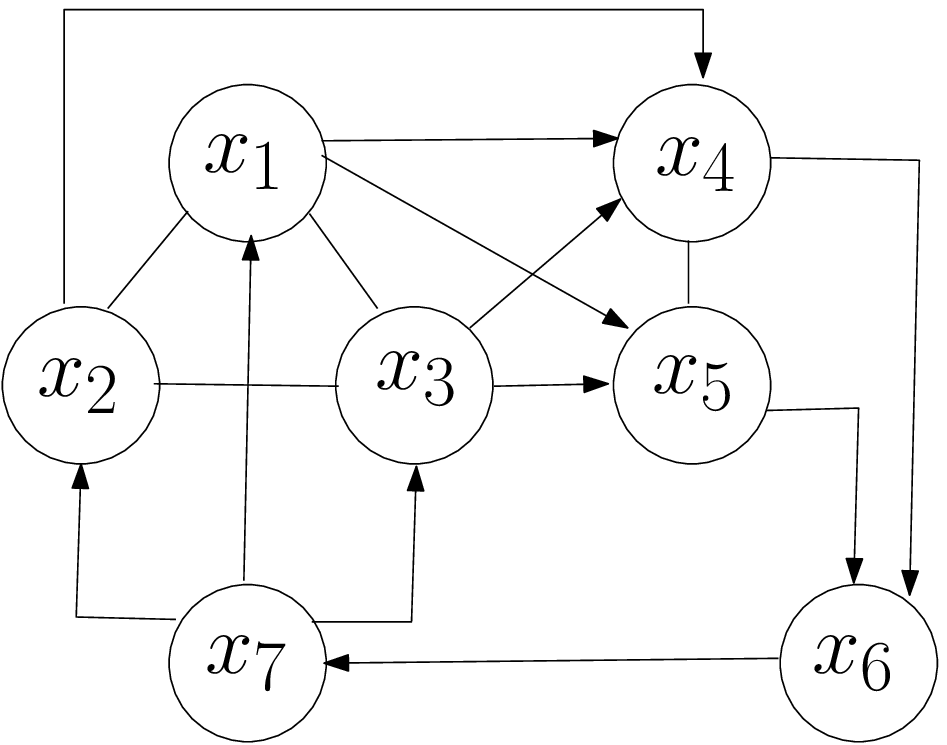}
\caption{}
\label{minrankfig3}
\end{figure}
\begin{figure}[h]
~~~~~~~~~~~~~~~~~~~~~~~~~~~~~~\includegraphics[scale=0.5]{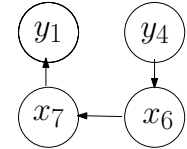}
\caption{Reduced side-information of $G$ given in Fig. \ref{minrankfig3}}
\label{minrankfig4}
\end{figure}
\end{example}

\begin{example}
\label{ex99}
Consider the index coding problem represented by the side-information graph $G$ given in Fig. \ref{minrankfig5}. From Construction I, the reduced side-information graph $G_R$ is given in Fig. \ref{minrankfig6}. In this example, the minrank of the graph $G$ is three and the minrank of graph $G_R$ is three. The index code for the index coding problem represented by $G_R$ is $\{y_1+y_4,y_4+x_6,x_6+x_7\}$ and the index code for the index coding problem represented by $G$ is $\{\underbrace{x_1+x_2+x_3}_{y_1}+\underbrace{x_4+x_5}_{y_4},~~\underbrace{x_4+x_5}_{y_4}+x_6,~~x_6+x_7\}$.
\begin{figure}[h]
~~~~~~~~~~~~~~~~~\includegraphics[scale=0.5]{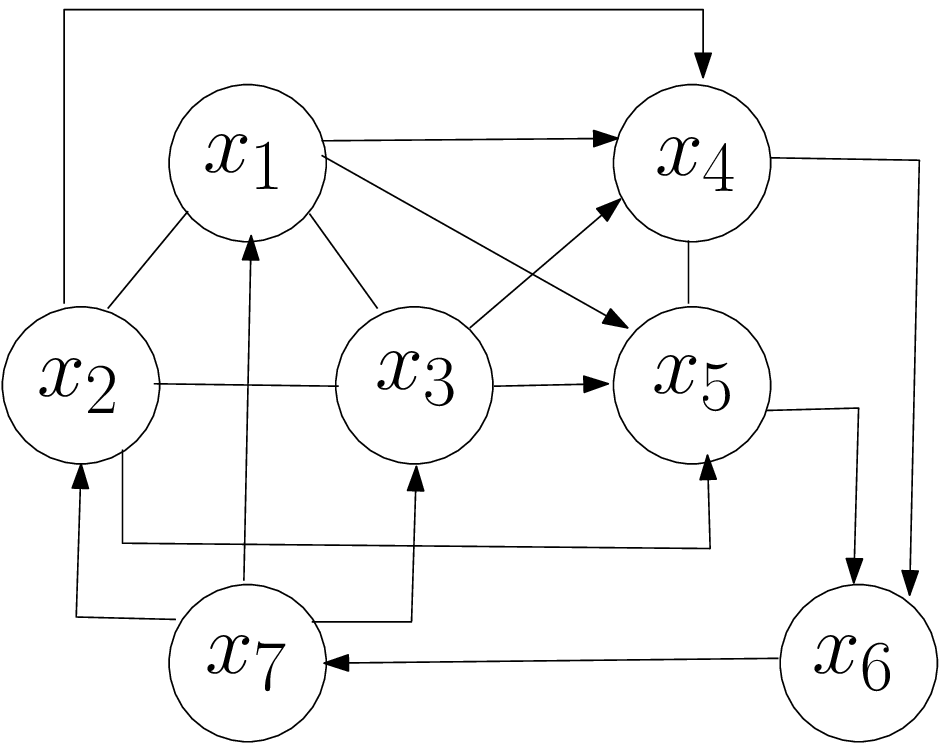}
\caption{}
\label{minrankfig5}
\end{figure}
\begin{figure}[h]
~~~~~~~~~~~~~~~~~~~~~~~~~~\includegraphics[scale=0.5]{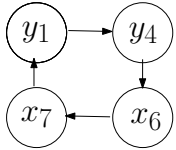}
\caption{Reduced side-information of $G$ given in Fig. \ref{minrankfig5}}
\label{minrankfig6}
\end{figure} 
\end{example}

In Example \ref{ex98} and Example \ref{ex99}, the minrank of side-information graph before and after reduction by using Construction I is same. However, the reduction method given in Construction I may not necessarily keep the rank same. For some graphs, the reduction procedure given in Construction I may increase the minrank of reduced side-information graph. Example \ref{ex100} given below is useful to understand this.

\begin{example}
\label{ex100}
Consider the index coding problem represented by the side-information graph $G$ given in Fig. \ref{minrankfig71}. In $G$, the vertices $\{x_1,x_2,x_3\}$ form a clique of size three and the vertices $\{x_4,x_5\}$ form a clique of size two. In the graph $G$, there exists a directed edge from every vertex in the clique $\{x_1,x_2,x_3\}$ to every vertex in the clique $\{x_4,x_5\}$. But, there does not exists an edge from every vertex in the clique $\{x_4,x_5\}$ to $x_6$. From Construction I, the reduced side-information graph $G_R$ is given in Fig. \ref{minrankfig81}. In this example, the minrank of the graph $G$ is nine, but the minrank of graph $G_R$ is ten.
\begin{figure}[h]
~~~~~~~~~~~~~~~~~\includegraphics[scale=0.5]{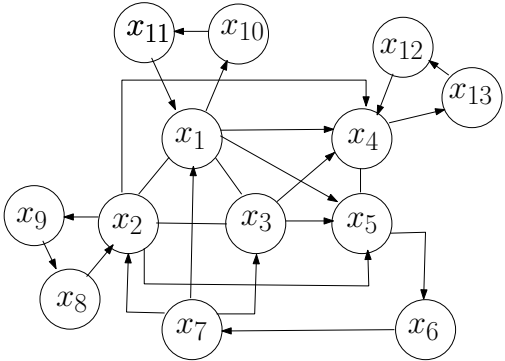}
\caption{Side-information graph $G$ with minrank nine}
\label{minrankfig71}
\end{figure}
\begin{figure}[h]
~~~~~~~~~~~~~~~~~~~~~~~~~~\includegraphics[scale=0.5]{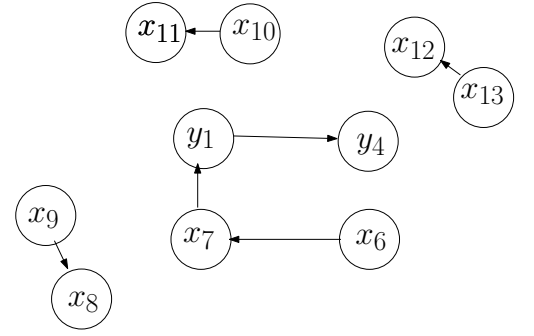}
\caption{Reduced side-information of $G$ given in Fig. \ref{minrankfig71}}
\label{minrankfig81}
\end{figure} 
\end{example}

In Theorem \ref{theorem2} and Construction I, it is assumed that the cliques in the graph are known. Note that finding a clique cover of a graph is an NP-hard problem. There exist various heuristic algorithms to find clique covers. In the Appendix, we give a heuristic algorithm to find the cliques by using binary operations on the adjacency matrix. There exist polynomial time cycle detection algorithms to check the cycle among a given set of vertices \cite{grapht}. Hence, given two cliques, one can find whether two cliques are cycle-free or not in polynomial time. 


\section{Code construction for groupcast index coding problems}
\label{sec4}
In this section, we give a method to convert a groupcast index coding problem into a single unicast index coding problem. This method, along with the other techniques given in this paper leads to a construction of index code for groupcast index coding problem.

\subsection{Converting groupcast ICP into single unicast ICP}
Consider a groupcast index coding problem with $K$ messages $\{ x_1,x_2,\ldots,x_K\}$ and a set of $m$ receivers $\{ R_1,R_2,\ldots,R_m\}$. Let $\mathcal{W}_k$ be the want-set and $\mathcal{K}_k$ be the side-information of receiver $R_k$ for $k \in [1:m]$. In groupcast index coding problem, there are no restrictions on want-set and side-information of each receiver. 


\begin{theorem}
\label{theorem4}
Consider a groupcast index coding problem with $K$ messages and $m$ receivers. Let $\gamma_k$ be the set of receivers wanting the message $x_k$ for $k \in [1:K]$ and 
\begin{align*}
\mathcal{\tilde{K}}_k=\bigcap_{\forall R_j \in \gamma_k} \mathcal{K}_j.
\end{align*}
\noindent Consider a single unicast index coding problem with $K$ messages $\{ x_1,x_2,\ldots,x_K\}$ and $K$ receivers $\{ \tilde{R}_1,\tilde{R}_2,\ldots,\tilde{R}_K\}$. The $k$th receiver $\tilde{R}_k$ wanting $x_k$ and having the side-information  $\mathcal{\tilde{K}}_k$. Then, any index code for this single unicast ICP is also an index code for the groupcast ICP.
\end{theorem}
\begin{proof}
Let $\mathfrak{C}$ be the index code for the converted single unicast index coding problem. From $\mathfrak{C}$, every receiver $\tilde{R}_k$ for $k \in [1:K]$ can decode its wanted message $x_k$. From the definition of $\mathcal{\tilde{K}}_k$, we have
\begin{align*}
\mathcal{K}_k \supseteq \bigcup_{\forall x_j \in \mathcal{W}_k} \mathcal{\tilde{K}}_j
\end{align*}
for $k \in [1:m]$. Hence, receiver $R_k$ for $k \in [1:m]$ can decode all its wanted messages in $\mathcal{W}_k$ from the given index code $\mathfrak{C}$.
\end{proof}


\subsection{Steps to construct index code for groupcast index coding problems}
In the following four steps, we give a heuristic approach to construct an index code for groupcast index coding problems. We refer the following four steps as Construction II in the rest of the paper. \\

\noindent
{\bf Construction II}

\begin{step}
Convert the given groupcast index coding problem into a single unicast index coding problem by using Theorem \ref{theorem4}.
\end{step}
\begin{step}
\label{step2}
Find the clique cover by using Algorithm \ref{algo1}. Reduce the given minrank problem into a smaller problem by using Construction I.
\end{step}
\begin{step}
\label{step3}
Find the cycle cover in the reduced minrank problem by using any cycle cover algorithm.
\end{step}
\begin{step}
Construct the index code by using the clique cover and cycle cover found in Step \ref{step2} and Step \ref{step3}.
\end{step}

The other method that can be used to construct index codes for groupcast problems is partition multicast \cite{bipartiate}. However, the partition multicast is NP-hard and requires higher field size. The field size required in partition multicast depends on the number of messages in a partition and the number of messages known to each receiver in the partition. Whereas, Construction II can be used to construct index code in polynomial time and this method is independent of field size. Note that both partition multicast and Construction II are suboptimal in the length of index code. 

\begin{example}
\label{ex10}
Consider a groupcast index coding problem with seven messages and ten receivers as given in Table \ref{table11}. The single unicast index coding problem corresponding to the groupcast index coding problem obtained from Theorem \ref{theorem4} is given in Table \ref{table21}. Cliques in this single unicast index coding problem can be found by using Algorithm \ref{algo1} and this single unicast index coding problem can be converted into a reduced index coding problem by using Construction I. The side-information graph of single unicast ICP given in Table \ref{table21} and its reduced side-information graph are shown in Fig. \ref{minrankfig7} and Fig. \ref{minrankfig8}. The minrank of $G_R$ and hence the minrank of $G$ is four. The index code for the index coding problem represented by $G_R$ is $\{y_1+x_3,x_3+y_4,y_4+x_6,x_6+x_7\}$ and the index code for the index coding problem represented by $G$ is $\{\underbrace{x_1+x_2}_{y_1}+x_3,~x_3+\underbrace{x_4+x_5}_{y_4},~\underbrace{x_4+x_5}_{y_4}+x_6,~x_6+x_7\}$.
\begin{table}[ht]
\centering
\setlength\extrarowheight{3.25pt}
\begin{tabular}{|c|c|c|c|}
\hline
$R_k$ & $\mathcal{W}_k$ & $\mathcal{K}_k$ \\
\hline
$R_1$ &$x_1,x_4$& $x_2,x_3,x_5,x_6$\\
\hline
$R_2$ &$x_1,x_5$& $x_2,x_3,x_4,x_6$\\
\hline
$R_3$ &$x_2,x_4$& $x_1,x_3,x_5,x_6$\\
\hline
$R_4$ &$x_4,x_7$& $x_1,x_2,x_5,x_6$\\
\hline
$R_5$ &$x_7$& $x_1,x_2$\\
\hline
$R_5$ &$x_3,x_6$& $x_4,x_5,x_7$\\
\hline
$R_6$ &$x_3$& $x_4,x_5$\\
\hline
$R_7$ &$x_5,x_7$& $x_1,x_2,x_4,x_6$\\
\hline
$R_8$ &$x_2,x_6$& $x_1,x_3,x_7$\\
\hline
$R_9$ &$x_6,x_1$& $x_2,x_3,x_7$\\
\hline
$R_{10}$ &$x_2,x_5$& $x_1,x_3,x_4,x_6$\\
\hline
\end{tabular}
\caption{~~}
\label{table11} 
\end{table}


\begin{table}[ht]
\centering
\setlength\extrarowheight{3.25pt}
\begin{tabular}{|c|c|c|c|}
\hline
$R_k$ & $\mathcal{W}_k$ & $\mathcal{K}_k$ \\
\hline
$\tilde{R}_1$ &$x_1$& $x_2,x_3$\\
\hline
$\tilde{R}_2$ &$x_2$& $x_1,x_3$\\
\hline
$\tilde{R}_3$ &$x_3$& $x_4,x_5$\\
\hline
$\tilde{R}_4$ &$x_4$& $x_5,x_6$\\
\hline
$\tilde{R}_5$ &$x_5$& $x_2,x_4,x_6$\\
\hline
$\tilde{R}_6$ &$x_6$& $x_7$\\
\hline
$\tilde{R}_7$ &$x_7$& $x_1,x_2$\\
\hline
\end{tabular}
\caption{Single unicast ICP obtained from Theorem \ref{theorem4} for the groupcast ICP given in Table \ref{table11}}
\label{table21}
\end{table}
\begin{figure}[h]
~~~~~~~~~~~~~~~~~\includegraphics[scale=0.5]{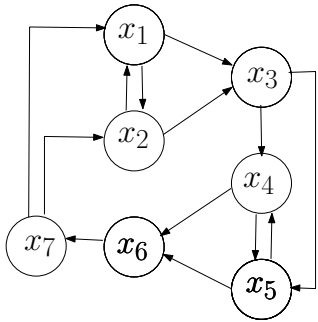}
\caption{Side-information graph of single unicast ICP given in Table \ref{table21}}
\label{minrankfig7}
\end{figure}
\begin{figure}[h]
~~~~~~~~~~~~~~~~~~~~~~~~~~\includegraphics[scale=0.5]{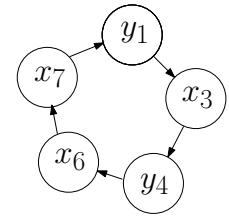}
\caption{Reduced side-information of $G$ given in Fig. \ref{minrankfig7}}
\label{minrankfig8}
\end{figure}
\end{example}
\begin{note}
For the groupcast index coding problem given in Table \ref{table11}, the index code length obtained by using partition multicast is five, whereas, by using Construction II, we can construct index code of length four as shown in Example \ref{ex10}.
\end{note}
\begin{example}
\label{ex11}
Some of the groupcast index coding problems in which the length of the index code given by Construction II is less than the length obtained from  partition multicast are given in Table \ref{table22}. In Table \ref{table22}, we use $K$,$m$,$l^*$ and $l_{PM}$ to denote number of messages, number of receivers, length of index code by using Construction II and length of index code by using partition multicast respectively. The minimum field size required to construct the index code is mentioned with the length of index code in both the methods. For the groupcast index coding problem given in S. No. 6 of Table \ref{table22}, the minimum field size required in partition multicast is $\mathbb{F}_{11}$, whereas Construction II gives the index code in $\mathbb{F}_2$ and the length of index code given by Construction II is one less than that of partition multicast.
\begin{table*}[t]
\centering
\setlength\extrarowheight{3.25pt}
\begin{tabular}{|c|c|c|c|c|c|c|c|}
\hline
S.No & $K$& $m$& $\mathcal{W}_k$ & $\mathcal{K}_k$ & Index Code& $l^*$ & $l_{PM}$\\
\hline
$1$ &$6$& $5$ & $\mathcal{W}_1=\{x_1,x_6\},\mathcal{W}_2=\{x_2,x_6\}$ & $\mathcal{K}_1=\{x_2,x_3,x_4\},\mathcal{K}_2=\{x_1,x_3,x_4\}$ & $\mathfrak{C}=\{x_1+x_2+x_3+x_4,$&$3$&$4$\\
~ &~& ~ & $\mathcal{W}_3=\{x_3\},\mathcal{W}_4=\{x_4\}$ &  $\mathcal{K}_3=\{x_1,x_2,x_4\},\mathcal{K}_4=\{x_5\}$ & $x_4+x_5,~x_6\}$&$(\mathbb{F}_2)$&$(\mathbb{F}_2)$\\
~ &~& ~ & $\mathcal{W}_5=\{x_5,x_6\}.$ & $\mathcal{K}_5=\{x_1,x_2,x_3\}.$ & ~&~&~\\
\hline
$2$ &$9$& $8$ & $\mathcal{W}_1=\{x_1,x_2\},\mathcal{W}_2=\{x_3,x_5\}$ & $\mathcal{K}_1=\{x_3,x_4,x_5\},$ & $\mathfrak{C}=\{x_1+x_5,$&$5$&$7$\\
~ && ~ & $\mathcal{W}_3=\{x_4,x_9\},\mathcal{W}_4=\{x_7\}$ &  $\mathcal{K}_2=\{x_2,x_4,x_7,x_8\},\mathcal{K}_3=\{x_1,x_6\},$ & $x_5+x_7+x_8,$&$(\mathbb{F}_2)$&$(\mathbb{F}_2)$\\
~ &~& ~ & $\mathcal{W}_5=\{x_4,x_8\},\mathcal{W}_6=\{x_6\}.$ & $\mathcal{K}_4=\{x_8,x_9\},\mathcal{K}_5=\{x_6,x_7,x_9\},$ & $x_7+x_8+x_9,$&~&~\\
~ &~& ~ & $\mathcal{W}_7=\{x_1,x_4\},\mathcal{W}_8=\{x_6,x_9\}.$ & $\mathcal{K}_6=\{x_2,x_3\},\mathcal{K}_7=\{x_5,x_6\},$ &$x_2+x_3+x_4,$&~&~\\
~ &~& ~ & ~ & $\mathcal{K}_8=\{x_1,x_2,x_3\},$ &$x_4+x_6\}$&~&~\\
\hline
$3$ &$7$& $8$ & $\mathcal{W}_1=\{x_1,x_7\},\mathcal{W}_2=\{x_3,x_6\}$ & $\mathcal{K}_1=\{x_3,x_6\},\mathcal{K}_2=\{x_2,x_4,x_7\}$ & $\mathfrak{C}=\{x_1+x_3,$&$4$&$5$\\
~ &~& ~ & $\mathcal{W}_3=\{x_2,x_6\},\mathcal{W}_4=\{x_4,x_7\}$ &  $\mathcal{K}_3=\{x_4,x_5,x_7\},\mathcal{K}_4=\{x_2,x_5,x_7\}$ & $x_2+x_3+x_4,$&$(\mathbb{F}_2)$&$(\mathbb{F}_2)$\\
~ &~& ~ & $\mathcal{W}_5=\{x_5\},\mathcal{W}_6=\{x_5,x_6\}.$ & $\mathcal{K}_5=\{x_1\},\mathcal{K}_6=\{x_1,x_7\}.$ & $x_2+x_4+x_5,\}$&~&~\\
~ &~& ~ & $\mathcal{W}_7=\{x_2\},\mathcal{W}_8=\{x_1\}.$ & $\mathcal{K}_7=\{x_4,x_5\},\mathcal{K}_8=\{x_3\}.$ & $x_6+x_7\}$&~&~\\
\hline
$4$ &$12$& $6$ & $\mathcal{W}_1=\{x_1,x_7,x_9\},$ & $\mathcal{K}_1=\{x_2,x_3,x_{10},x_{11},x_{12}\}$ & $\mathfrak{C}=\{x_1+x_2+x_3,$&$6$&$9$\\
~ &~& ~ & $\mathcal{W}_2=\{x_1,x_{11},x_{12}\},$ & $\mathcal{K}_2=\{x_2,x_3,x_5,x_6,x_8,x_{10}\}.$ &$x_5+x_6+x_{12},$&$(\mathbb{F}_2)$&$(\mathbb{F}_2)$\\
~ &~& ~ & $\mathcal{W}_3=\{x_2,x_{12}\},$ &  $\mathcal{K}_3=\{x_1,x_3,x_5,x_6\}$ & $x_5+x_6+x_7,$&~&~\\
~ &~& ~ & $\mathcal{W}_4=\{x_3,x_5,x_{10}\},$ & $\mathcal{K}_4=\{x_4,x_6,x_7,x_8,x_{11}\}.$ &$x_3+x_4,x_8+x_9,$&~&~\\
~ &~& ~ & $\mathcal{W}_5=\{x_4,x_6\},\mathcal{W}_6=\{x_3,x_8\}.$ & $\mathcal{K}_5=\{x_1,x_2,x_5,x_7\},\mathcal{K}_6=\{x_4,x_9\}.$ & $x_9+x_{10}+x_{11}\}$&~&~\\
\hline
$5$ &$12$& $6$ & $\mathcal{W}_1=\{x_1,x_7,x_9\},$ & $\mathcal{K}_1=\{x_2,x_3,x_{10},x_{11},x_{12}\}$ & $\mathfrak{C}=\{x_1+x_2+x_3,$&$6$&$8$\\
~ &~& ~ & $\mathcal{W}_2=\{x_1,x_{11},x_{12}\},$ & $\mathcal{K}_2=\{x_2,x_3,x_5,x_6,x_8,x_{10}\}.$ &$x_5+x_6+x_{12},$&$(\mathbb{F}_2)$&$(\mathbb{F}_{13})$\\
~ &~& ~ & $\mathcal{W}_3=\{x_2,x_{12}\},$ &  $\mathcal{K}_3=\{x_1,x_3,x_5,x_6\}$ & $x_5+x_6+x_7,$&~&~\\
~ &~& ~ & $\mathcal{W}_4=\{x_3,x_5,x_{10}\},$ & $\mathcal{K}_4=\{x_4,x_6,x_7,x_8,x_{11}\}.$ &$x_3+x_4,x_8+x_9,$&~&~\\
~ &~& ~ & $\mathcal{W}_5=\{x_4,x_6\},\mathcal{W}_6=\{x_3,x_8\}.$ & $\mathcal{K}_5=\{x_1,x_2,x_5,x_7\},$ & $x_9+x_{10}+x_{11}\}$&~&~\\
~ &~& ~ & ~ & $\mathcal{K}_6=\{x_4,x_5,x_6,x_9\}.$ & ~&~&~\\
\hline
$6$ &$10$& $8$ & $\mathcal{W}_1=\{x_1,x_2,x_{10}\},$ & $\mathcal{K}_1=\{x_3,x_4,x_5\},$ & $\mathfrak{C}=\{x_1+x_5,$&$6$&$7$\\
~ && ~ & $\mathcal{W}_2=\{x_3,x_5,x_{10}\},$ &  $\mathcal{K}_2=\{x_2,x_4,x_7,x_8\},$ & $x_5+x_7+x_8,$&$(\mathbb{F}_2)$&$(\mathbb{F}_{11})$\\
~ &~& ~ & $\mathcal{W}_3=\{x_4,x_9\},\mathcal{W}_4=\{x_7\},$ & $\mathcal{K}_3=\{x_1,x_6,x_{10}\},\mathcal{K}_4=\{x_8,x_9,x_{10}\},$ & $x_7+x_8+x_9,$&~&~\\
~ &~& ~ & $\mathcal{W}_5=\{x_4,x_8\},\mathcal{W}_6=\{x_6\},$ & $\mathcal{K}_5=\{x_6,x_7,x_9\},\mathcal{K}_6=\{x_2,x_3,x_{10}\},$ &$x_2+x_3+x_4,$&~&~\\
~ &~& ~ & $\mathcal{W}_7=\{x_1,x_4\},\mathcal{W}_8=\{x_6,x_9\}.$ & $\mathcal{K}_7=\{x_5,x_6,x_{10}\},\mathcal{K}_8=\{x_1,x_2,x_3\},$ &$x_4+x_6,~x_{10}\}$&~&~\\
\hline
\end{tabular}
\caption{Some instances of the groupcast index coding problem where the length of the index code given by Construction II is less than the length obtained from partition multicast. }
\label{table22}
\end{table*}

\end{example}
\section{conclusion and discussions}
\label{sec5}
In this paper, we give a method to find some of the minrank-non-critical edges in the side-information graph. We give a simple heuristic method to find the clique cover by using binary operations on adjacency matrix. We presented a method to address groupcast index coding problems. It is interesting to analyze more properties of minrank and design algorithms to compute the minrank of a side-information graph in a more efficient way.

\section*{Acknowledgement}
This work was supported partly by the Science and Engineering Research Board (SERB) of Department of Science and Technology (DST), Government of India, through J.C. Bose National Fellowship to B. Sundar Rajan.

\begin{appendix}
\begin{center}
A HEURISTIC ALGORITHM TO FIND A CLIQUE COVER
\end{center}

Birk \textit{et al.} \cite{ISCO} proposed least difference greedy (LDG) clique cover algorithm to find the cliques in side-information graph. Kwak \textit{et al.} \cite{eldg} improved the LDG algorithm by proposing extended least difference greedy (ELDG) clique cover algorithm to find the cliques in side-information graph. However, the way the cliques are found in both LDG and ELDG algorithms depends also on the directed edges in the side-information graph which do not contribute to cliques. In this section, we give a method for the heuristic search of the cliques in the side-information graph by using binary operations on the adjacency matrix.

LDG and ELDG algorithms use fitting matrix to find cliques, whereas the algorithm presented in this paper use the adjacency matrix.  LDG and ELDG algorithms are given below for continuity and completeness of presentation.

\subsection{Least Difference Greedy Clique-Cover Algorithm}
The distance between two rows in a fitting matrix is defined as the sum of the inter-entry distance, where the inter-entry distance $d$ is defined as
\begin{align}
\label{interrow}
&d(0,0)=d(1,1)=d(*,*)=0,\\&
d(0,*)=d(1,*)=1,~~\text{and}~~d(0,1)=\infty.
\end{align}
Based on the defined inter-row distance, the LDG algorithm finds the minimum distance among all possible pairs of two rows in a fitting matrix and merges two rows with the minimum inter-row distance, and then iterates this procedure until all inter-row distances become infinite.

\subsection{Extended Least Difference Greedy Clique-Cover Algorithm}
The ELDG algorithm is based on the observation that the minrank of the fitting matrix $\mathbf{A}$ and $\mathbf{A}^\mathsf{T}$ is same. In ELDG algorithm, the LDG algorithm is applied on both the rows and columns of the fitting matrix and the algorithm chooses the row/column merging requiring less *'s.

In both LDG and ELDG algorithms, the inter-row distance given in \eqref{interrow} depends on directed edges which do not contribute to any cliques. Hence, some of the directed edges present in the graph can mislead the row/column merging in ELDG algorithm. Consider the side information graph $G$ given in Fig. \ref{fig21}. The row inter entry distances (denoted by $d_r(i,j)$ for the $i$th and $j$th rows) and the column inter entry distances (denoted by $d_c(i,j)$ for the $i$th and $j$th columns) are given in Table \ref{table1} and \ref{table2}. In Table \ref{table1} and \ref{table2}, the inter row and inter column distances is the minimum between rows $2$ and $4$ and columns $3$ and $5$. Hence, in $G$, the ELDG algorithm merges the two sets of rows $(2,4)$ and $(3,5)$. Hence, the rank reduction by ELDG algorithm is two, whereas, one can reduce the rank by three by merging the rows $(1,2,3)$ and $(4,5)$.

    \begin{figure}[h]
        \includegraphics[scale=0.45]{minrankcase60}
  \caption{}
  \label{fig21}
\end{figure}

\begin{figure}
{
$\left[\begin{array}{ccccccccccccccc}
  1 & * & * & 0 & 0 & * & * & * & * & 0 & 0 & 0 & 0 & 0 & 0\\
  * & 1 & * & * & 0 & 0 & 0 & 0 & 0 & * & * & 0 & 0 & 0 & 0\\
  * & * & 1 & 0 & * & 0 & 0 & 0 & 0 & 0 & 0 & * & * & 0 & 0\\
  0 & * & 0 & 1 & * & * & * & 0 & 0 & * & * & 0 & 0 & 0 & 0\\

  0 & 0 & * & * & 1 & 0 & 0 & * & * & 0 & 0 & * & * & 0 & 0\\
  0 & * & 0 & 0 & 0 & 1 & 0 & 0 & 0 & 0 & 0 & 0 & 0 & 0 & 0\\
  0 & * & 0 & 0 & 0 & 0 & 1 & 0 & 0 & 0 & 0 & 0 & 0 & 0 & 0\\
  0 & 0 & * & 0 & 0 & 0 & 0 & 1 & 0 & 0 & 0 & 0 & 0 & 0 & 0\\
  0 & 0 & * & 0 & 0 & 0 & 0 & 0 & 1 & 0 & 0 & 0 & 0 & 0 & 0\\
  0 & 0 & * & 0 & * & 0 & 0 & 0 & 0 & 1 & 0 & 0 & 0 & 0 & 0\\
  0 & 0 & * & 0 & * & 0 & 0 & 0 & 0 & 0 & 1 & 0 & 0 & 0 & 0\\
  0 & * & 0 & * & 0 & 0 & 0 & 0 & 0 & 0 & 0 & 1 & 0 & 0 & 0\\
  0 & * & 0 & * & 0 & 0 & 0 & 0 & 0 & 0 & 0 & 0 & 1 & 0 & 0\\
  0 & * & 0 & * & 0 & 0 & 0 & 0 & 0 & 0 & 0 & 0 & 0 & 1 & 0\\
  0 & 0 & * & 0 & * & 0 & 0 & 0 & 0 & 0 & 0 & 0 & 0 & 0 & 1\\
  \end{array}\right].$
}
\caption{Fitting matrix of side-information given in Fig. \ref{fig21}}
\end{figure}

\setlength\extrarowheight{3.25pt}
\begin{table}[h]
\centering
\begin{tabular}{|c|c|c|}
\hline
$d_r(1,2)=9$ & $d_r(1,3)=9$ & $d_r(1,j)=\infty$ for $j \in [4:13]$ \\
\hline
$d_r(2,3)=8$ & $d_r(2,4)=7$ & $d_r(1,j)=\infty$ for $j \in [5:13]$ \\
\hline
$d_r(3,4)=\infty$ & $d_r(3,5)=7$ & $d_r(3,j)=\infty$ for $j \in [6:13]$ \\
\hline
$d_r(4,5)=11$ & $d_r(4,6)=\infty$ & $d_r(1,j)=\infty$ for $j \in [7:13]$ \\
\hline
\end{tabular}
\caption{Row inter-entry distances of $G$ given in Fig. \ref{fig21}}
\label{table1}
\end{table}

\setlength\extrarowheight{3.25pt}
\begin{table}[h]
\centering
\begin{tabular}{|c|c|c|}
\hline
$d_c(1,2)=8$ & $d_c(1,3)=8$ & $d_c(1,j)=\infty$ for $j \in [4:13]$ \\
\hline
$d_c(2,3)=11$ & $d_c(2,4)=7$ & $d_c(1,j)=\infty$ for $j \in [5:13]$ \\
\hline
$d_c(3,4)=\infty$ & $d_c(3,5)=7$ & $d_c(1,j)=\infty$ for $j \in [6:13]$ \\
\hline
$d_c(4,5)=8$ & $d_c(4,6)=\infty$ & $d_c(1,j)=\infty$ for $j \in [7:13]$ \\
\hline
\end{tabular}
\caption{Column inter-entry distances of $G$ given in Fig. \ref{fig21}}
\label{table2}
\end{table}

The Hadamard product of two matrices $\mathbf{A}$ and $\mathbf{B}$ is denoted by $\mathbf{A}\circ \mathbf{B}$ and defined by $(\mathbf{A}\circ \mathbf{B})(i,j)=\mathbf{A}(i,j)\cdot \mathbf{B}(i,j)$. That is, the Hadamard product of two matrices is the element wise binary and operation of its elements.

\begin{lemma}
\label{lemma1}
Let $\mathbf{A}$ be the adjacency matrix of the graph $G$, define $\mathbf{B}=\mathbf{A}^{\mathsf{T}} \circ \mathbf{A}$. The matrix $\mathbf{B}$ is a symmetric matrix and represents the location of undirected edges in $G$.
\end{lemma}
\begin{proof}
If there exists a directed edge from $x_i$ to $x_j$, we have $\mathbf{A}(i,j)=1$, $\mathbf{A}(j,i)=0$,  $\mathbf{A}^{\mathsf{T}}(i,j)=0$ and $\mathbf{A}^{\mathsf{T}}(j,i)=1$. Hence, we have
\begin{align*}
&\mathbf{B}(i,j)=\mathbf{A}(i,j)^{\mathsf{T}}\circ \mathbf{A}(i,j)=0, \ and \\ &\mathbf{B}(j,i)=\mathbf{A}(j,i)^{\mathsf{T}}\circ \mathbf{A}(j,i)=0.
\end{align*}

If there exists an undirected edge from $x_i$ to $x_j$, we have $\mathbf{A}(i,j)=1$, $\mathbf{A}(j,i)=1$,   $\mathbf{A}^{\mathsf{T}}(i,j)=1$ and $\mathbf{A}^{\mathsf{T}}(j,i)=1$. Hence, we have
\begin{align*}
&\mathbf{B}(i,j)=\mathbf{A}^{\mathsf{T}}(i,j)\circ \mathbf{A}(i,j)=1, \ and \\ &\mathbf{B}(j,i)=\mathbf{A}^{\mathsf{T}}(j,i)\circ \mathbf{A}(j,i)=1.
\end{align*}
This completes the proof.
\end{proof}

By using Lemma \ref{lemma1}, in Algorithm \ref{algo1}, we give a heuristic method to find the cliques in the side-information graph. In the $l$th iteration, Algorithm \ref{algo1} finds cliques of size two in the adjacency matrix given by $(l-1)$th iteration. Note that every row (or column) of the adjacency matrix in $(l-1)$th iteration corresponds to a clique detected in $(l-1)$th iteration.

                \begin{algorithm}
                        \caption{Heuristic algorithm to find a clique cover of a side-information graph}
                        \label{algo1}
                        \begin{algorithmic}[1]
                                \item $\mathbf{A}$ is the adjacency matrix of graph $G$ of order $V(G)\times V(G)$.
                                \item $C_i=\{i\},i=1,2,\ldots,V(G)$, $C=\phi$ and $n=V(G)$.
                                \item $\mathbf{B}=\mathbf{A}^{\mathsf{T}} \circ \mathbf{A}$ and $S=\phi,D=\phi,E=\phi$.
                                \If {$\mathbf{B}$ is a zero matrix}, goto step 5
                                \begin{itemize}
                                \item[\footnotesize{4.1:}]Let $R=\{k_1,k_2,\ldots,k_n\}$ be the set of row indices in the ascending order of their Hamming weights ($\mathbf{A}$ is a $n \times n$ matrix).
                                \item[\footnotesize{4.2:}] $i =1$ and $t=0$.
                                \item[\footnotesize{4.3:}] \textbf{if} $wt(k_i=0)$ \textbf{then},
                                \item[\footnotesize{4.4:}]  $E=E\cup\{k_i\}$ and $i \gets i+1$, repeat step 4.3.

                                        \item[\footnotesize{4.4:}] \textbf{else if} $B(k_i,k_j)=1$ (WLOG we assume $k_j>k_i$)
                                        \begin{itemize}
                                        \item[$\bullet$] $S \gets S \cup \{(k_i,k_j)\}$, $C_{k_i}=C_{k_i}\cup C_{k_j}$ and $C_{k_j}=\phi$.

                                        \item[$\bullet$] $R \gets R \backslash \{k_i,k_j\} $.
                                        \item[$\bullet$] $i \gets i+1$ and $t=t+1$.
                                        \item[$\bullet$] \textbf{if} {$i<n$} \textbf{then}, repeat Step 4.4.

                                        \end{itemize}
                              
                                \item[\footnotesize{4.5:}] $S_1=S$.
                                \item[\footnotesize{4.6:}] \textbf{if} {$(r_i,r_j)\in S_1$} \textbf{then},
                                \begin{itemize}
                                \item[$\bullet$] $A(r_i,:)=A(r_i,:)\circ A(r_j,:)$.
                                \item[$\bullet$] $S_1=S_1 \backslash \{(r_i,r_j)\}$.
                                \item[$\bullet$] \textbf{if} {$S_1\neq \phi$} \textbf{then}, repeat Step 4.6.

                                \end{itemize}

                                \item[\footnotesize{4.7:}] $S_2=S$.
                                \item[\footnotesize{4.8:}] \textbf{if} {$(r_i,r_j)\in S_2$} \textbf{then},
                                \begin{itemize}
                                \item[$\bullet$] $A(:,r_i)=A(:,r_i)\circ A(:,r_j)$.
                                \item[$\bullet$] $S_2=S_2 \backslash \{(r_i,r_j)\}$.
                                \item[$\bullet$] $D=D\cup \{j\}$
                                \item[$\bullet$] \textbf{if} {$S_2\neq \phi$} \textbf{then}, repeat Step 4.8.

                                \end{itemize}
                                \item[\footnotesize{4.9:}] Delete the rows and columns in the $n \times n$ matrix $\mathbf{A}$ whose indices are present in the set $D$ and $E$. Output the $C_j$s for every $j \in D$.
                                \item[\footnotesize{4.10:}] \textbf{do} $\{$
                                \item[\footnotesize{4.11:}]Let $d_{k_i}$ be the number of rows deleted in $\mathbf{A}$ above $k_i$ (rows whose indices less than $k_i$). $C_{k_i-d_{k_i}}=C_{k_i}$ and $t=t-1$.
                                \item[\footnotesize{4.12:}] \textbf{while} $\{t>0\}$.
                                \item[\footnotesize{4.13:}] $n=n-\mid D \mid -\mid E \mid.$
                                \item[\footnotesize{4.14:}] Goto Step 3.
                                \end{itemize}
                                \EndIf
                                .
                        \end{algorithmic}
                \end{algorithm}

\subsection{Algorithm description}
In the first iteration, the algorithm finds cliques of size two in $G$. That is, after the first iteration every row (or column) in $\mathbf{A}$ represents a vertex in $G$ or a clique of size two in $G$. Similarly, after second iteration, every row (or column) in $\mathbf{A}$ represents a vertex in $G$ or a clique of size 2/3/4 in $G$. The iterations continue until there exist no more combinable rows and columns in the matrix $\mathbf{A}$ obtained from the previous iteration (this condition is equivalent to $\mathbf{B}=0$). In Step 4.1, the row indices are arranged in the ascending order of their Hamming weights to ensure that the algorithm first combines the vertices which are connected with less number of undirected edges (the vertices which are connected with more number of undirected edges have more options for combining, hence they are treated later in the order).

The purpose of different sets used in Algorithm \ref{algo1} are given below.
\begin{itemize}
\item The set $S$ keeps track of two tuples of row indices that form a clique in the present iteration
\item The set $D$ keeps track of the indices of rows which would be merged with the other rows. The rows and columns corresponding to the indices present in $D$ would be deleted before going to the next iteration.
\item The set $E$ keeps track of the rows in $\mathbf{B}$ whose Hamming weight is zero. The rows and columns corresponding to the indices present in $E$ would be deleted before going to the next iteration.
\item The sets $C_{k_1},C_{k_2},\ldots,C_{k_n}$ give the cliques identified till the present iteration.
\end{itemize}

\subsection{Computational complexity of Algorithm \ref{algo1}}
By using Algorithm \ref{algo1}, all cliques of size two can be found by using $n^2$ binary $AND$ operations. If the graph consists of a clique of maximum size $l$, after atmost $l$ iterations the matrix $\mathbf{B}$ becomes zero. Therefore the computational complexity of finding all $m$-cliques is less than $ln^2$ binary $AND$ operations.

For the given graph $G$, Algorithm \ref{algo1} finds all cliques of size two and combines them in the first iteration. Similarly, in $l$th iteration, Algorithm \ref{algo1} finds and combines all cliques of size two in the graph corresponding to the adjacency matrix obtained from $(l-1)$th iteration. However, Algorithm \ref{algo1} can be modified to combine only one clique of size two in every iteration. This might yield less number of cliques in some side-information graphs at the cost of more iterations.

\begin{example}
For the side-information graph given in Fig. \ref{fig21}, in the first iteration, Algorithm \ref{algo1} finds all the cliques of size two and merges the cliques $\{4,5\}$ and $\{1,2\}$. In the second iteration, Algorithm \ref{algo1} combines the cliques $\{1,2\}$ and $\{3\}$. After second iteration, $\mathbf{B}$ is zero and Algorithm \ref{algo1} outputs the cliques $\{1,2,3\}$ and $\{4,5\}$.
\end{example}

\end{appendix}


\begin{thebibliography}{160}
\bibitem{ISCO}
Y. Birk and T. Kol, ``Coding-on-demand by an informed-source (ISCOD) for efficient broadcast of different supplemental data to caching clients", in IEEE \textit{Trans. Inf. Theory,}, vol. 52, no.6, pp.2825-2830, June 2006.
\bibitem{OMIC}
L Ong and C K Ho, ``Optimal Index Codes for a Class of Multicast Networks with Receiver Side Information'', in \textit{Proc. IEEE ICC}, 2012, pp. 2213-2218.
\bibitem{ICSI}
Z. Bar-Yossef, Z. Birk, T. S. Jayram, and T. Kol, ``Index coding with side information", in IEEE \textit{Trans. Inf. Theory,}, vol. 57, no.3, pp.1479-1494, Mar. 2011.

\bibitem{ECIC}
S.~H. Dau, V. Skachek, and Y.~M. Chee, ``Error Correction for Index Coding With Side Information", in IEEE \textit{Trans. Inf. Theory,}, vol. 59, no.3, pp.1517-1531, Mar. 2013.
\bibitem{TSG}
M. Tahmasbi, A. Shahrasbi and A. Gohari,  ``Critical Graphs in Index Coding",  in IEEE \textit{Journal in selected areas of Communications}, vol. 33, no.2, pp.225-235, Feb. 2015.

\bibitem{minrank2}
B. Recht. M. Fazel and P. Parrilo, ``Guaranteed minimum-rank solutions on linear matrix equations via nuclear norm minimization," SIAM Review, vol. 52, no. 3, pp. 471-501, 2010.


\bibitem{matrix}
H. Esfahanizadeh, F. Lahouti, and B. Hassibi, ``A matrix completion approach to linear index coding problem,'' ITW 2014, Australia, November 2014.
\bibitem{eldg}
S. Kwak, J. So, and Y. Sung, ``A Extended Least Difference Greedy Clique-Cover Algorithm for Index Coding,'' ISIT 2014, Australia, 2014.
\bibitem{cycle1}
Koshy. T, ``Discrete mathematics with applications,'' Academic Press, 2003.

\bibitem{cycle2}
M. Awais and J. Qureshi, ``Efficient coding for unicast flows in opportunistic wireless networks,'' IET Communications., 2017, Vol. 11, Iss. 4, pp. 584–592.

\bibitem{randomgraph}
B. Bollobas and P. Erdosw, ``Cliques in random graphs", Mathematical Proceedings of the Cambridge 1976.

\bibitem{bipartiate}
A. S. Tehrani, A. G. Dimakis, and M. J. Neely, ``Bipartite index coding,” in \textit{Proc. IEEE ISIT 2012}, pp. 2246–2250.

\bibitem{grapht}
J. A. Bondy and U. S. R. Murty, ``Graph theory", Graduate texts in mathematics 244, Springer.

\bibitem{randomgraph2}
G. R. Grimmett and C. J. H Mcdiarrnd, ``On colouring random graphs", Math. Proc. Cambridge Philos. Soc. 77, 1975, pp. 313-324.

\end{thebibliography}
\end{document}